\documentclass[journal]{IEEEtranTIE}

\usepackage{graphicx}
\usepackage{cite}
\usepackage{picinpar}
\usepackage{amsmath}
\usepackage{url}
\usepackage{flushend}
\usepackage[latin1]{inputenc}
\usepackage{colortbl}
\usepackage{soul}
\usepackage{multirow}
\usepackage{pifont}
\usepackage{color}
\usepackage{alltt}
\usepackage[hidelinks]{hyperref}
\usepackage{enumerate}
\usepackage{siunitx}
\usepackage{breakurl}
\usepackage{epstopdf}
\usepackage{pbox}
\usepackage{physics}
\usepackage{amssymb}

\usepackage{amsthm}
\usepackage{xcolor}
\usepackage{epsfig} % for postscript graphics files
\usepackage{graphics} % for pdf, bitmapped graphics files

\newtheorem{theorem}{{\bf{Theorem}}}

\usepackage{subfigure}
\usepackage{bm}
\usepackage[export]{adjustbox}
\begin{document}
	\title{	
		A Nonlinear Proportional Integral Disturbance Observer and Motion Control Technique for Permanent Magnet Synchronous Motors}
	
	\author{
		\vskip 1em
		
		Yong Woo Jeong, \emph{IEEE Student Member},
		and Chung Choo Chung$^\dag$, \emph{IEEE Member}
		
		\thanks{
			
%			Manuscript received Month xx, 2xxx; revised Month xx, xxxx; accepted Month x, xxxx.
%			This work was supported in part by the xxx Department of xxx under Grant  (sponsor and financial support acknowledgment goes here).
			%
			Yong Woo Jeong is with the Electrical Engineering Department, University of Hanyang, Seoul, 04763, Korea (e-mail: elecjeong@hanyang.ac.kr).
			Chung Choo Chung is with the Div. of Electrical and Biomedical Engineering, University of Hanyang, Seoul, 04763, Korea (e-mail: cchung@hanyang.ac.kr).
		}
		\thanks{$\dag$: Corresponding Author}
	}
	
	\maketitle
	
	\begin{abstract}
		In this paper, we present a Nonlinear-Proportional Integrator (N-PI) disturbance observer (DOB) to enhance the motion tracking of the performance of a surface-mounted Permanent Magnet Synchronous Motor (SPMSM) in rapidly speed varying regions. By presenting an N-PI-DOB for load torque estimation with torque modulation technique, we show that the tracking error dynamics of angular position/velocity are coupled with tracking errors of currents loop and estimation errors. After analyzing disturbances of currents tracking error dynamics, we design the N-PI-DOB and Lyapunov-based nonlinear currents controller to enhance the motion tracking performances. With these N-PI-DOBs and motion controllers, we analyze the stability of motion tracking error dynamics and estimation error dynamics. We experimentally perform a comparative study with/without the N-PI-DOB to verify the effectiveness of the proposed method in the condition of the unknown load torque and rapidly speed-varying.
\end{abstract}


%\markboth{IEEE TRANSACTIONS ON INDUSTRIAL ELECTRONICS}%
{}

\definecolor{limegreen}{rgb}{0.2, 0.8, 0.2}
\definecolor{forestgreen}{rgb}{0.13, 0.55, 0.13}
\definecolor{greenhtml}{rgb}{0.0, 0.5, 0.0}

\section{Introduction}
%

%%%%%%%%%%%%%%%%%%%%%%%%%%%%%%%%%%%%%%%%%%%%%%%
%%%%%%%%%%%%%%% Paragraph 1 %%%%%%%%%%%%%%%%%%%
% Paragraph 1-1]   for describing why the PMSM Precision Motion Tracking is important.
% Paragraph 1-2]   for the Principle of PMSM motion control.
%%%%%%%%%%%%%%%%%%%%%%%%%%%%%%%%%%%%%%%%%%%%%%%
%%%%%%%%%%%%%%%%%%%%%%%%%%%%%%%%%%%%%%%%%%%%%%%
\IEEEPARstart{P}{recision} motion control of a Permanent Magnet Synchronous Motor (PMSM) has drawn attention from the manufacturing industry because of its efficiency and compact structure~\cite{lee2017lpv}.
To achieve precision position/velocity control performances, the construction of two connected systems, i.e., electro-mechanical system, is essential.
One is an outer-loop control system related to the motor shaft dynamics, i.e., mechanical system. The other is an inner-loop control system related to the inverter dynamics, i.e., an electrical system~\cite{jakovljevic2021control}.
The outer-loop control system generates desired torque reference once desired motion reference of the motor shaft is given.
Then, the inner-loop control system, such as either  a field orient control~(FOC) or direct torque control (DTC), generates phase voltages to achieve motor torque as the desired torque reference~\cite{casadei2002foc}.
In the outer/inner-loop control system, a feedback controller along with a feed-forward controller is utilized.
However, disturbance terms in outer/inner-loop control systems, caused by  un-modeled nonlinear dynamics, hinder PMSMs from having robust position/velocity tracking performances. 

For that reason, a linear DOB with a feedback controller for PMSMs is presented in~\cite{zhu2000speed}, and it ensures the local exponential stability by taking parameter fluctuations and the external load torque as the lumped disturbance.
Further, an extended state-based DOB with the sliding mode motion controller has been applied at the servo-tracking system and showed comparative studies with/without the DOB~\cite{lee2016integral}.
In addition, a high-order terminal sliding mode-based DOB was implemented in~\cite{feng2012high} to estimate the mechanical parameters of PMSMs.
Later, a super-twisting sliding mode-based disturbance observer for PMSMs was developed in~\cite{hou2020composite}, showing its outperformance by presenting the comparative study.
In~\cite{hou2020gpio}, disturbances in inverter dynamics are considered lumped disturbances. Moreover, a generalized proportional integral observer (GPIO) was introduced to estimate the lumped disturbance.
Although the outer-loop controllers presented in~\cite{lee2016integral,feng2012high,zhu2000speed,hou2020composite} estimate and compensate disturbance, only related with the outer-loop system, they do not consider and compensate disturbances, related to inner-loop control system caused by nonlinear properties of inverter dynamics. 
The disturbances existing in the inner-loop control system can hinder not only performances in steady-state speed operating regions but also in transient speed operating regions.

%  

%%%%%%%%%%%%%%%%%%%%%%%%%%%%%%%%%%%%%%%%%%%%%%%
%%%%%%%%%%%%%%% Paragraph 3 %%%%%%%%%%%%%%%%%%%
% Paragraph 3 ]   Disturbances in the inner-loop control system.
% A Brief Introduction of research of Nonlinear Estimator and compensation Techniques.
%%%%%%%%%%%%%%%%%%%%%%%%%%%%%%%%%%%%%%%%%%%%%%%
%%%%%%%%%%%%%%%%%%%%%%%%%%%%%%%%%%%%%%%%%%%%%%%
% 

%	\cite{kim2012microstepping} have applied an Proportional-Integrator DOB for inner-loop control system of PM stepper motor with variable structure controller.
%
%
%
To compensate for the disturbances of the inner-loop control system, \cite{liu2017robust} have applied a sliding mode disturbance compensation and predictive current controller. However, the structure of DOB has limits in that it does not contain the integration term of estimation errors which has been reported that enhance the performance of the control/estimation system.
%
%
%	Furthermore, to suppress the current harmonics of permanent magnet synchronous linear motors, a control method that combines a vector resonant controller and an active disturbance rejection control is proposed in~\cite{wang2019combined}.
%	
%
Recently, a Lyapunov-based nonlinear inner-loop controller is presented to enhance the motion tracking performance of a surface-mounted Permanent Magnet Synchronous Motor~(SPMSM) in~\cite{jeong2020nonlinear}, and it shows the comparative study of the motion tracking performances with the conventional inner-loop controllers.
Although this paper shows that the velocity tracking performance can be enhanced by enhancing the inner-loop control performance, there was no disturbance-observer-based feed-forward compensation technique in both outer/inner-loop control systems.
% 
	%%%%%%%%%%%%%%%%
	 
%
For that reason, in this paper, we present a nonlinear Proportional-Integrator~(N-PI) outer/inner-loop disturbance observer~(DOB) and precision motion controller to enhance the speed tracking performances in  rapidly speed varying regions for   SPMSM.
The main contributions of this paper can be summarized as  follows:
\begin{itemize}
  \item[1)] A new N-PI DOB is proposed to estimate the outer/inner-loop disturbance, i.e., unknown load torque, nonlinearity in the inverter system. Then, a Lyapunov redesign-based nonlinear currents controller  is presented.
  \item[2)]  After deriving the tracking error dynamics of SPMSM, we proved that the disturbances in the inner-loop system and current tracking errors converge to the bounded ball. Further, we show that the load torque estimation error and velocity tracking error converge to bounded value.
  \item[3)] Finally, with the experimental validation results, we show that the proposed N-PI-DOBs in the outer/inner loop system guarantee the uniform motion tracking performances not only in steady-state region but also in rapidly speed-varying regions.
%
%  \item[1)] A new N-PI DOB is proposed to estimate the unknown load torque. WIth the torque modulation, we derive the tracking error dynamics of SPMSM Then, a Lyapunov redesign-based nonlinear currents controller with the proposed DOB is presented.
%  \item[2)] With the N-PI DOB and motion controller, we proved that current tracking errors become zeros in finite time. Further, we show that velocity tracking error converges to bounded value.
%%
%  \item[3)]  Finally, with the experimental validation results, we show that the proposed nonlinear DOB in the outer/inner loop system guarantees the uniform motion tracking performances not only steady-state region but also transient one. 

\end{itemize}

To describe the pre-mentioned contribution of this paper, in section~\ref{sec:section2}, we describe the modeling of SPMSM, torque modulation, and N-PI-DOB for load torque estimation. After then, we derive tracking error dynamics of SMPSM and discuss the possible disturbances which hinder convergence of motion tracking errors. In Section~\ref{sec:section3}, we present the N-PI-DOB for inner-loop disturbance estimation and stability analysis of tracking error dynamics of SPMSM.  After presenting the experimental validation results, a conclusion will follow. 

\section{ SPMSM Modelling and Torque Modulation and N-PI-DOB for Load Torque Estimation}
\label{sec:section2}

This section describes SPMSM dynamic equation and the outer-loop controller, including torque modulation and nonlinear proportional integral disturbance observer~(N-PI-DOB) for load torque estimation.
Based on the torque modulation and N-PI-DOB, we derive tracking error dynamics of SPMSM and discuss possible disturbance terms in inner-loop control system which may hinder the regulation of motion tracking errors.

\subsection{Electro-Mechanical Dynamics of SPMSM}

To begin with, let us define $\theta, \omega, i_\alpha,$ and $i_\beta$ as state variables of SPMSM where
$\theta$  is the rotor's angular position~(rad),
$\omega$ is the rotor's angular velocity~(rad/s), and $i_\alpha$, $i_\beta$ are currents~(A) of $\alpha$$\beta$ frame.
Then, the dynamics of SPMSM can be represented as%~\cite{kim2017electric} 
\begin{align}
	\begin{split}
		\label{eq:PMSM_dynamics}
		\dot{\theta} &= \omega\\			
		\dot \omega  &= - \frac{B}{J}\omega +\frac{1}{J}{\tau _m}   - \frac{1}{J}{\tau _L}   \\
		{\dot i}_\alpha  &=  { - \frac{R}{L}{i_\alpha } + \frac{P\Phi}{L} \sin(P\theta) \omega  + \frac{1}{L}{v_\alpha }   }     \\
		{\dot i}_\beta  &= - \frac{R}{L}{i_\beta } - \frac{P\Phi}{L} \cos(P\theta)  \omega + \frac{1}{L}{v_\beta }   \\
	\end{split}
\end{align}
where	$B$ is the viscous friction coefficient~(N$\cdot$m$\cdot$s/rad), and $J$ is the inertia of the motor~(kg$\cdot$m$^2$).
$v_\alpha$ and $v_\beta$ are voltages~(V). $\tau_L$ is the load torque (Nm).
$R$ and $L$ are the resistance of the phase winding~($\Omega$) and the inductance of the phase winding~(H).
$\Phi$ is the rotor's magnetic flux~(Wb).
$P$ is the number of pole pairs.
%
%Here, for the simplicity of analysis, we assume that $B$ and $J$ are known. Further, we assume that we know the nominal values, ${R},{L},{\Phi}$, of parameters, $R, L, \Phi$. Then we can represent the parameters as
%	\[
%	R ={R}+\Delta{R},~ L ={L}+\Delta{L}, ~\Phi ={\Phi}+\Delta{\Phi}\]
%	where $\Delta{R},\Delta{L}$, and $\Delta{\Phi}$ are uncertainties of the parameters.
% 
%
For simplicity of notations, let us define $ {S} : =\sin(P\theta)$, and $ {C}:= \cos(P\theta)$. Given $i_\alpha$ and $i_\beta$, the motor torque of   SPMSM, ${\tau _m}$, can be represented as 
	\begin{equation}
	{\tau _m} =  - \frac{3}{2}P\Phi  {S} {i_\alpha } + \frac{3}{2}P\Phi  {C} 	{i_\beta }.
	\label{eq:Tm}
\end{equation}
\subsection{Torque Modulation in the presence of Load Torque }
In this subsection, we introduce torque modulation in the presence of the unknown load torque. 
Given desired position/velocity reference $\theta_d$ and $\omega_d$, and measurements $\theta$ and $\omega$, the torque modulation generates desired current references, $i_\alpha^d$ and $i_\beta^d$.
%Let us define desired current reference as $i_\alpha^d$ and $i_\beta^d$.
%
To begin with, let us define the estimation states for the outer-loop control system such as, $\hat{\omega}$ and $\hat{\tau}_L$.  $\hat{\omega}$ is a velocity estimation state which is only used for load torque estimation. $\hat{\tau}_L$ is a load torque estimation state.
The detail structure of $\hat{\omega}$ and $\hat{\tau}_L$ will be presented in~Sec.~\ref{sec:section2_C}.  
%	With the states in~(\ref{eq:PMSM_dynamics}), the desired reference such as $\theta^d$, $\omega^d$, $i_{\alpha}^d$ and $\i_{\beta}^d$, and estimation states $\hat{\omega}$ and $\hat{\tau}_L$,
Given SPMSM states in~(\ref{eq:PMSM_dynamics}), the desired references  and estimation states,	let us define tracking errors and estimation errors such as 
\begin{equation}
	\begin{split}
		\label{eq:error_state}
		%		\begin{matrix}
		%			e_\theta = \theta^d-\theta, & e_\alpha = i_\alpha ^d - {i_\alpha }, &  \tilde{\omega} = \omega-\hat{\omega}
		%			\\
		%			e_\omega = {\omega ^d} - \omega, & e_\beta  = i_\beta ^d - {i_\beta }, &  \tilde{d}_{\tau}  = (\tau_L+\Delta\tau_m ) -\hat\tau_L\\
		%		\end{matrix}
		&e_\theta = \theta^d-\theta, ~ ~~ e_\omega = {\omega ^d} - \omega,  ~  \tilde{\omega} = \omega-\hat{\omega}
		\\
		&e_\alpha = i_\alpha ^d - {i_\alpha }, ~~ e_\beta  = i_\beta ^d - {i_\beta },~   \tilde{\tau}_L  = {\tau}_L -\hat{\tau}_L
	\end{split}
\end{equation}
where $e_\theta$ is the rotor position tracking errors and $e_\omega$ is the rotor velocity tracking errors. $e_\alpha$ and $e_\beta$ are currents tracking errors,	$\tilde{\omega}$ is  the estimation error of rotor velocity and  $\tilde{\tau}_L$  is the estimation error of unknown load torque. 
With these notations, we can design a desired torque reference, $\tau_m^d$, such as
\begin{align}
	\begin{split}
		\label{eq:Tm_d}
		\tau_m^d &:= ({J}\dot {\omega}^d+{B}\omega^d +k_\theta e_\theta +k_\omega e_\omega +\hat{\tau}_{L} )\\
		&=- \frac{3}{2}P{\Phi}  {S}  {i_\alpha^d } + \frac{3}{2}P{\Phi}  {C}   {i_\beta^d }\\
%			&=- \frac{3}{2}P{\Phi}  {S}  {e_\alpha } + \frac{3}{2}P{\Phi}  {C}   {e_\beta }+\tau_m
	\end{split}
\end{align}
where $k_\theta$ and $k_\omega$ are control gains for the stabilization of the mechanical system.
From the characteristics of trigonometric functions,   $i_\alpha^d$ and $i_\beta^d$  are computed as% using the following equations:
\begin{equation}
	\label{eq:NTM}
	\begin{split}
		i_\alpha^d &:=-\frac{2{\tau_m^d}}{3P{\Phi}} {S}, \,\,
		i_\beta^d :=\frac{2{\tau_m^d}}{3P{\Phi}} {C} .\\
	\end{split}
\end{equation}

\subsection{Nonlinear-PI-DOB for Outer-Loop Control System}
\label{sec:section2_C}
In this subsection, we present a Nonlinear PI-based Disturbance Observer~(N-PI-DOB) for outer-loop control system to estimate $\tau_L$.
Suppose that mechanical dynamics of SPMSM is given by~(\ref{eq:PMSM_dynamics}), and both $\theta$ and $\omega$ are measurable. 	%
Then, with parameters $l_{p,\tau}$, $l_{i,\tau}$,  and $\tilde{\omega}^{\max}>0$,  the  N-PI-DOB for load torque estimation is as follow: 
\begin{equation}
	\begin{split} 
		\dot {\hat{\omega}}   &=   
  - \frac{{B}}{{J}}\hat\omega+\frac{1}{{J}}{\tau _m^d} -\frac{1}{{J}}\hat{\tau}_{L} \\
		\hat{\tau}_L &=   - \mu_\tau(\tilde{\omega}) - l_{i,\tau} \int \tilde{\omega} d\tau .
	\end{split}
	\label{eq:tau_hat}
\end{equation}
The nonlinear function $\mu_\tau(\tilde{\omega})$ is %%given by 
$
	\mu_\tau(\tilde{\omega})=l_{p,\tau} \frac{\tilde{\omega} }{ (\tilde{\omega}/\tilde{\omega}^{\max})^2+1 }
$
where $l_{p,\tau}$ is the estimation gain related to the peak value of the $\mu_\tau(\tilde{\omega})$ as depicted in~Fig.~\ref{fig:dot_hat_tau_L}.
The derivative of $\hat{\tau}_L$ with respective to time is represented as follow:
\begin{equation*}
	\begin{split} 
		\dot{\hat{\tau}}_L  
		&=        -\frac{d\mu_\tau(\tilde{\omega})}{d\tilde{\omega}} \frac{d\tilde{\omega}}{dt} -l_{i,\tau} \tilde{\omega}.
		\label{eq:dot_hat_tau_L}
	\end{split}
\end{equation*}
	By defining $\partial \mu_\tau(\tilde{\omega}) :=\frac{d\mu_\tau(\tilde{\omega})}{d\tilde{\omega}}$, $ {k}_m := \frac{3}{2}P {\Phi}$, from~(\ref{fig:dot_hat_tau_L}), (\ref{eq:error_state}) and (\ref{eq:tau_hat}), the estimation error dynamics of $\tilde{\omega}$ and $\tilde{\tau}_L$ becomes \begin{equation}
	\begin{split}
		\dot {\tilde{\omega}}   &= -\frac{{B}}{{J}}\tilde\omega    +  \frac{{k}_m }{J}
		\begin{bmatrix}
		 S & - C
		\end{bmatrix}
		\begin{bmatrix}
			e_\alpha \\ e_\beta
		\end{bmatrix}-\frac{1}{{J}} \tilde{\tau}_{L} \\
		\dot{\tilde{\tau}}_L &=  \left\{ -\partial  \mu_\tau(\tilde{\omega}) \frac{{B}}{{J}}+l_{i,\tau}\right\}\tilde{\omega}  - \frac{\partial  \mu_\tau(\tilde{\omega})}{{J}}\tilde{\tau}_L \\
		&  ~~~~~~~~~~~~~~ +      \frac{\partial\mu_\tau(\tilde{\omega}){k}_m }{J} \begin{bmatrix}
		 S  &  -C
		\end{bmatrix}
		\begin{bmatrix}
			e_\alpha \\ e_\beta
		\end{bmatrix}+ \dot{\tau}_L.
	\end{split}
	\label{eq:estim_dynamics}
\end{equation}   
For simplicity of notation, let us define the mechanical tracking error   states $\mathbf{e}_{m}$, currents tracking error states $\mathbf{e}_{\alpha\beta}$, and the estimation error states for outer-loop control system $\mathbf{d}_{\tau}$ as 
\begin{equation}
	\mathbf{e}_{m}=\begin{bmatrix} e_\theta & e_\omega\end{bmatrix}^T, \mathbf{e}_{\alpha\beta}=\begin{bmatrix} e_\alpha & e_\beta\end{bmatrix}^T, \mathbf{d}_{\tau}=\begin{bmatrix} \tilde{\omega} & \tilde{\tau}_L \end{bmatrix}^T.
	\label{eq:Errors}
\end{equation}	 
With these notations,  (\ref{eq:estim_dynamics}) can be represented as 
\begin{equation}
	\label{eq:dtau_dynamics}
	\dot{\mathbf{d}}_\tau  = A_{\tau}(\tilde{\omega})\mathbf{d}_\tau    +   G_\tau(\tilde{\omega})\mathbf{e}_{\alpha\beta} +B_\tau \dot{\tau}_L  \\
\end{equation}
where $A_{\tau}(\tilde{\omega})  
=   A_{\tau0} +\partial\mu_\tau(\tilde{\omega}) {A_{\tau1} }$, $B_\tau = \begin{bmatrix} 0 & 1\end{bmatrix}^T$
$G_\tau(\tilde{\omega} )  = \frac{{k}_m}{J}
		\begin{bmatrix}
			S   & -C \\
			\partial\mu(\tilde{\omega})  S   & -\partial\mu(\tilde{\omega}) C 
		\end{bmatrix}$,
 $A_{\tau0}=  
{\begin{bmatrix}
		-\frac{{B}}{{J}} &   -\frac{1}{{J}} \\
		l_{i,\tau} &  0 \\
\end{bmatrix}}$, 
and $A_{\tau1}={\begin{bmatrix}
		0 & 0\\
		-\frac{{B}}{{J}} &  - \frac{1}{{J}} \\
\end{bmatrix}}$. 
The value of $\dot{\tau}_L$ is related with the variance of injected $\tau_L$. So, in practice, without loss of generality we can assume that there exists $\delta_\tau$ such as
\begin{equation}
	\label{eq:delta_tau}
	%{\delta}^{max}_i = \sup_{\dot{d}_i \in D_i}{ \norm{\dot{d}_i}_{\infty}}.
	{\delta}_\tau = \sup_{t \in [0, \infty )} { \norm{\dot{\tau}_L}}.
\end{equation}
\subsection{Tracking Error Dynamics of SPMSM with N-PI-DOB and Torque Modulation}

To analyze the stability of motion tracking control with the torque modulation and N-PI-DOB and to design the inner-loop controller, in this subsection, we derive the tracking error dynamics of SPMSM.
By subtracting~(\ref{eq:Tm}) and (\ref{eq:Tm_d}), the following relationship can be derived
\begin{align}
	\begin{split}
		\label{eq:Tm_2} 
  \tau_m & =   \tau_m^d +   k_m ( {S}  {e_\alpha } -   {C}   {e_\beta }).\\ 
	\end{split}
\end{align}
By plugging equations~(\ref{eq:PMSM_dynamics}),~(\ref{eq:Tm_d})  and (\ref{eq:Tm_2}) into time derivative of (\ref{eq:error_state}), tracking error dynamics  can be derived as 
\begin{align}
	\begin{split}
		\label{eq:ded}
		\dot e_\theta =& e_\omega\\
		\dot e_\omega =&  - \frac{k_\theta}{J} e_\theta -  \frac{B+k_\omega }{J}  e_\omega
		- \frac{{k}_m }{J} (S  e_\alpha - C e_\beta)  +\frac{1}{{J}} \tilde{\tau}_L\\
		\dot e_\alpha =&  \dot{i}_\alpha ^d+ \frac{R}{L} i_\alpha	- \frac{1}{L}P {\Phi} S \omega 
		-  \frac{1}{L} v_\alpha     
		%+r_\alpha
		%
		\\
		{\dot e}_\beta   =&  \dot{i}_\beta ^d +  \frac{R}{L} i_\beta   +  \frac{1}{L}P {\Phi}  C  \omega  -  \frac{1}{L}{v_\beta }   .
		%+r_\beta
		\\
	\end{split}
\end{align}
Tracking error dynamics shows that $e_\alpha$, $e_\beta$ and  $\tilde{\tau}_L$  are related to the dynamics of   $e_\omega$.
Therefore, designing $v_\alpha$ and $v_\beta$ to regulate $e_\alpha$ and $e_\beta$ to be zero is essential.
Let us define voltage references, $v_\alpha^d$ and $v_\beta^d$, commanded to inverter systems as
\begin{equation}
	\begin{split}
		v_\alpha^d = Ri_\alpha^d-P\Phi S \omega -u_\alpha,~~ v_\beta^d = Ri_\beta^d + P\Phi C \omega-u_\beta
	\end{split}
\label{eq:vab_d}
\end{equation} 
where $u_\alpha$ and $u_\beta$ are feedback control inputs which will be discussed in Sec.~\ref{sec:section3}.
Due to the inverter nonlinearity, there exists voltage difference  between actual phase voltages and desired voltage reference, and it can be represented as   
\begin{equation}
	\begin{split}
		v_\alpha &= v_\alpha^d+ e_{v\alpha},~~~ v_\beta=v_\beta^d+ e_{v\beta}.
	\end{split}
\label{eq:v_error}
\end{equation}
where $e_{vj}$    are the voltage differences between $v_j^d$ and $v_j$, $j\in{\left\{\alpha,\beta\right\}}$. 
Further, from~(\ref{eq:ded}), we see that $\dot{i}_\alpha^d$ and $\dot{i}_\beta^d$ are coupled with the dynamics of $e_\alpha$ and $e_\beta$, and they can be disturbances for convergence of $e_\alpha$ and $e_\beta$. 
By considering these possible disturbance terms for inner-loop control system, let us define inner-loop disturbances as
\begin{equation}
	d_\alpha:={   -{L}\dot{i}_\alpha ^d    + e_{v\alpha}  }, ~~d_\beta:={    - {L}\dot{i}_\beta ^d   + e_{v\beta} }.
	\label{eq:dab}
\end{equation}
% $d_\alpha:={   -{L}\dot{i}_\alpha ^d    + e_{v\alpha}  }$, $d_\beta:={    - {L}\dot{i}_\beta ^d   + e_{v\beta} }$.
%
By plugging~(\ref{eq:vab_d}), (\ref{eq:v_error}) and (\ref{eq:dab}) into (\ref{eq:ded}), the we can represent SPMSM error dynamics  such as
\begin{align}
	\begin{split}
		\label{eq:ded2}
		\dot e_\theta =& e_\omega\\
		\dot e_\omega =&  - \frac{k_\theta}{{J}} e_\theta -  \frac{k_\omega+{B}}{{J}}  e_\omega
		- \frac{{k}_m }{J}(S  e_\alpha - C e_\beta) +\frac{1}{{J}} \tilde{\tau}_{L}\\
		\dot e_\alpha =&
		-  \frac{ {R}}{ {L}} e_\alpha   +  \frac{1}{ {L}}{u_\alpha }  -  \frac{1}{{L}}
		{d_\alpha}\\
		%%  \\
		%%
		{\dot e}_\beta   =&
		-  \frac{{R}}{{L}} e_\beta   +  \frac{1}{{L}}{u_\beta  }  - \frac{1}{{L}}
		{d_\beta}.\\
	\end{split}
\end{align}  
For simplicity of notations, let us define $\mathbf{u}_{\alpha\beta} =\begin{bmatrix} u_\alpha  & u_\beta  \end{bmatrix}^T$ and 
$\mathbf{d}_{\alpha\beta}=\begin{bmatrix} d_\alpha & d_\beta \end{bmatrix}^T$.
With (\ref{eq:Errors}), we can represent~(\ref{eq:ded2}) as
\begin{figure}[!t]
	\centering	
	\subfigure[  The graph of ${\mu}(x)=      l_{p,x}  \frac{   x   } {       (x/x^{\max})^2     +1     }$. ]
	{\includegraphics[width=0.97\linewidth ]{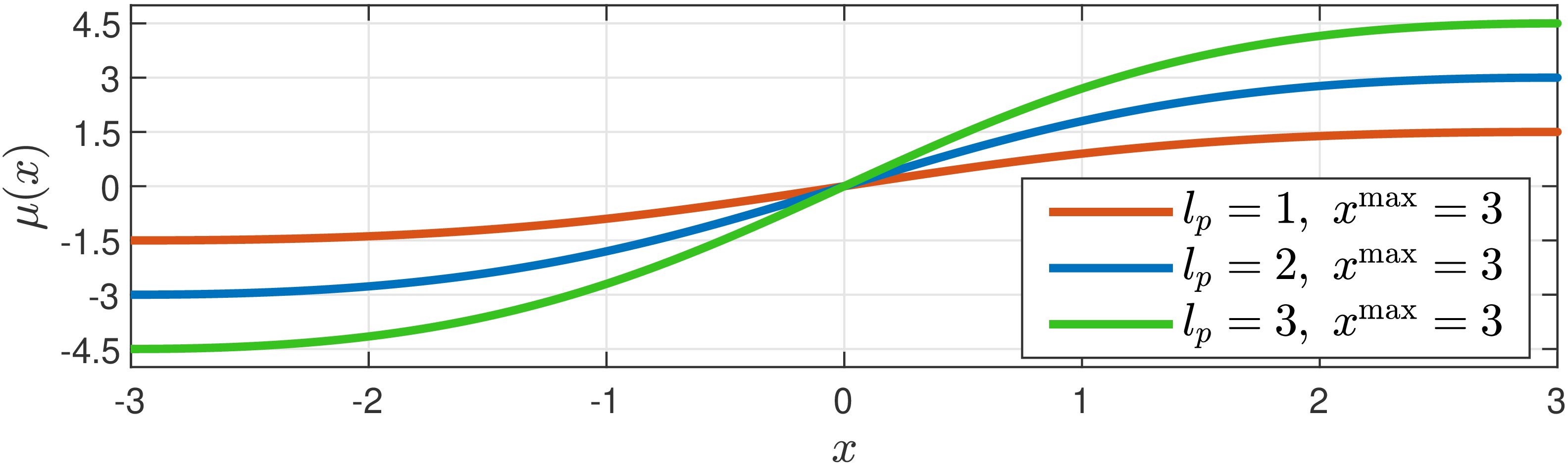}
	}
	\subfigure[  The graph of $\partial {\mu}(x)=     l_{p,x} \frac{ -(x/x^{\max})^2+1 } {     ( (x/x^{\max})^2     +1 )^2  }$. ]
	{\includegraphics[width=0.97\linewidth ]{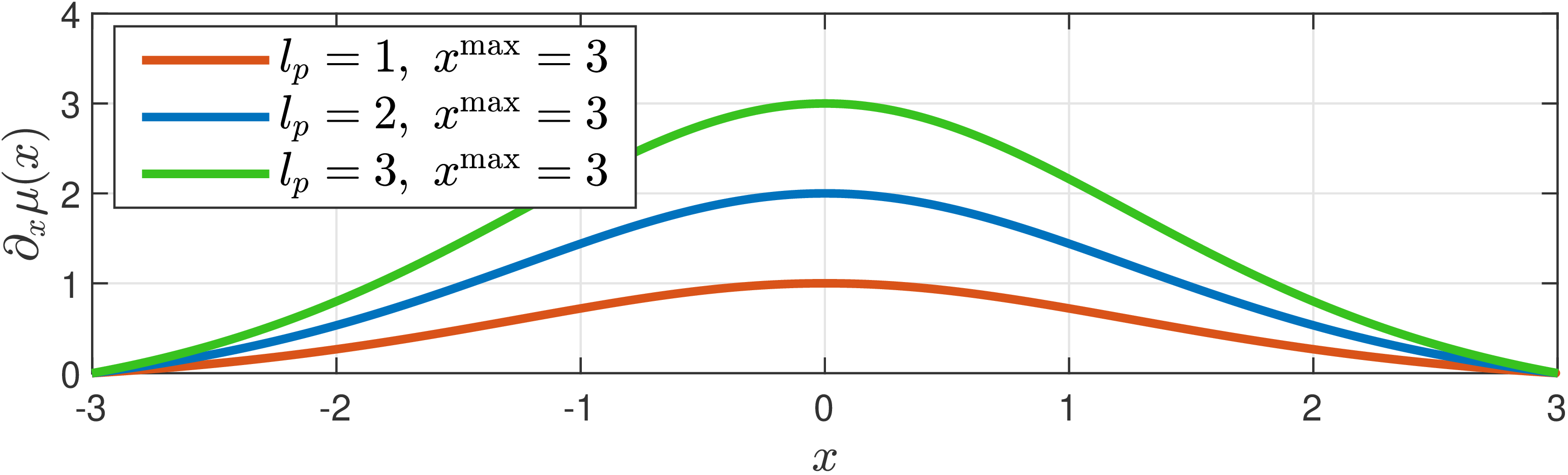}
		\label{fig:partial_Mu}
	}
	\caption{   The graph of nolinear function  $\mu(x)$ and $\partial_x{\mu}(x)$}
	\label{fig:dot_hat_tau_L}
\end{figure}
\begin{equation}
	\begin{split}
		\dot {\mathbf{e}}_{m} &= A_m \mathbf{e}_{m} +G_m\mathbf{e}_{\alpha\beta} +B_m \mathbf{d}_{\tau}\\
		\dot{\mathbf{d}}_\tau  &= A_{\tau}(\tilde{\omega})\mathbf{d}_\tau    +   G_\tau(\tilde{\omega})\mathbf{e}_{\alpha\beta} +B_\tau \dot{\tau}_L  \\
%			%%
		\dot{\mathbf{e}}_{\alpha\beta}  &=  A_{\alpha\beta} {\mathbf{e}}_{\alpha\beta} + B_{\alpha\beta} ({\mathbf{u}}_{\alpha\beta}  -  {\mathbf{d}}_{\alpha\beta})\\
	\end{split}
\label{eq:ded_final}
\end{equation}
where
\begin{equation*}
	\begin{split}
		& A_m =\begin{bmatrix}
			0 & 1\\
			-\frac{k_\theta}{{J}} &	-\frac{k_\omega+{B}}{{J}}
		\end{bmatrix},
		G_m = -\frac{{k}_m }{J}\begin{bmatrix}
			0	& 0  \\
		 S &  -C\\
		\end{bmatrix}, \\
		&B_m = \begin{bmatrix}
			0 &0		\\
			0 &\frac{1}{J}		\\
		\end{bmatrix},
		A_{\alpha\beta} = \begin{bmatrix}
			-\frac{{R}}{{L}} & 0		\\
			0 & 	-\frac{{R}}{{L}}		\\
		\end{bmatrix}, 
		B_{\alpha\beta} = \begin{bmatrix}
		 \frac{1}{{L}} & 0		\\
			0 & 	 \frac{1}{{L}}		\\
		\end{bmatrix}.
	\end{split}
\end{equation*}
%
%Based on this tracking error dynamics, we are going to design N-PI-DOB for inner-loop control system and Lyapunov currents controller. Then, we discuss the stability of N-PI-DOB and controller of SPMSM.   
Based on this tracking error dynamics, we are going to regulate $\mathbf{e}_{\alpha\beta}$ by designing N-PI-DOB for inner-loop control system and Lyapunov currents controller. Then, we discuss the stability of the motion tracking controller of SPMSM.

\begin{figure*}[!t]
	\centering	
	{
		\includegraphics[width=0.96\linewidth ]{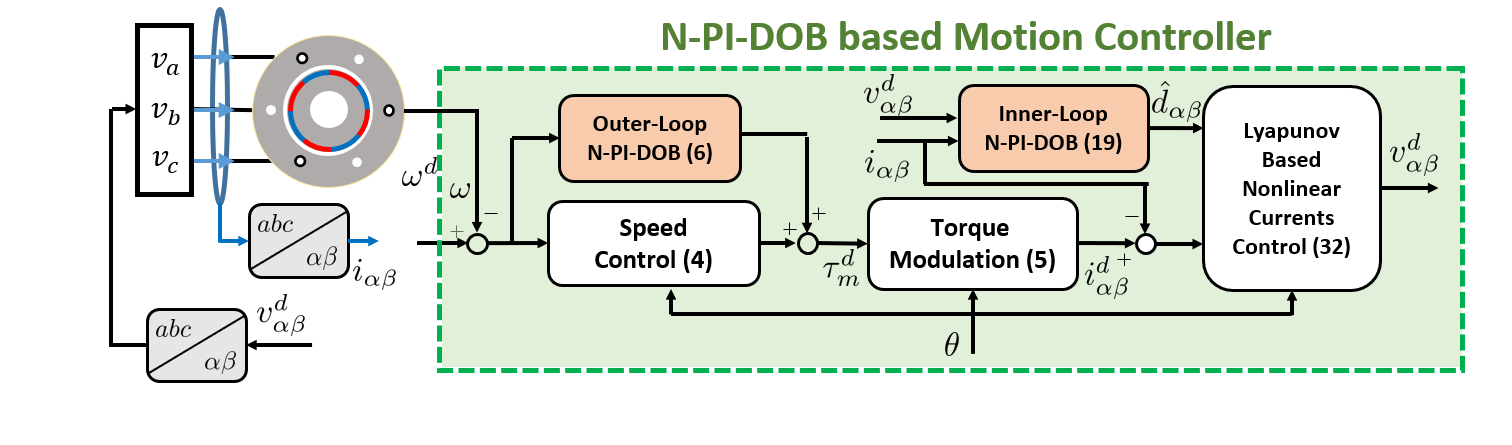}
	}
	\caption{  Schematic of the proposed N-PI-DOB-based velocity control structure of SPMSM. }
	\label{fig:N_PI_DOB_Control}
\end{figure*} 

\section{N-PI-DOB for Inner-Loop Control System, Nonlinear Currents Controller and Stability Analysis} 
\label{sec:section3}
%
%To achieve uniform motion tracking performance of $\theta$ and $\omega$ of SPMSM, in this section, we present a nonlinear-PI disturbance observer for inner-loop control system. Then, we briefly introduce a Lyapunov based nonlinear currents controller.
In this section, we present a nonlinear-PI disturbance observer for inner-loop control system. Then, we briefly introduce a Lyapunov based nonlinear currents controller.
Then, we perform the stability analysis of the proposed motion controller of SPMSM.
%  

%%%%%%%%%%%%%%%%%%%%%%%%%%%%%%%%%%%%%%%%%%%%%%%%%%%%%%%%%%%%
\subsection{Nonlinear-PI-DOB for Inner-Loop Control System}
\label{sec:N-PI-DOB_InnerLoop}
In this subsection, we present a Nonlinear PI-based Disturbance Observer~(N-PI-DOB) for inner-loop control system to estimate ${d}_{\alpha}$ and ${d}_{\beta}$ in~(\ref{eq:dab}).
Let us define estimation states of $e_\alpha,e_\beta$ as  $\hat{e}_j,~ j\in\{\alpha,\beta\}$ and estimation states of $d_\alpha,d_\beta$  as $\hat{d}_j,~ j\in\{\alpha,\beta\}$. Furthermore, let us define the estimation error of $e_j$ and  $d_j$ such as
%For simplicity of notations, let us define estimation errors of $e_\alpha,e_\beta$ as  $\tilde{e}_j,~ j\in\{\alpha,\beta\}$ and estimation errors of $d_\alpha,d_\beta$  as $\tilde{d}_j,~ j\in\{\alpha,\beta\}$, such as
\begin{equation}
	\begin{split}
		&\tilde{e}_j = e_j - \hat{e}_j, ~~~ \tilde{d}_j = d_j - \hat{d}_j,~~~ j \in \left\{ \alpha, \beta \right\}
	\end{split}
	\label{eq:ia_tilda}
\end{equation}
where, $\tilde{e}_j$ is the estimation error of  $e_j$,  $\tilde{d}_j$ is the estimation error of inner-loop disturbance $d_j$.   
%
%Suppose that tracking error dynamics of SPMSM is given by~(\ref{eq:ded2}), and $\theta$, $\omega$, and $i_j$ are measurable.
%
%,  we can design a N-PI-DOB for load torque estimation as follow: 
 Then, with parameters  $l_{p,e}$, $l_{i,e}$,  and $\tilde{e}^{\max}>0$,  we can design the N-PI-DOB for inner-loop control system as
\begin{equation}
	\begin{split}
		\dot{\hat{e}}_{j } &=   - \frac{{R}}{{L}} \hat{e}_j   + \frac{1}{{L}} u_j       - \frac{1}{{L}} \hat{d}_j \\
%		\hat{d}_{j } &=    -l_{\chi0}  \frac{\tilde{e}_{j }}{\tilde{e}_{i }^2/2+\epsilon_{\chi }} -l_{\chi1} \int \tilde{e}_{i } d\tau \\
		% 
		\hat{d}_j &=   - \mu_e(\tilde{e}_j) - l_{i,e} \int \tilde{e}_j d\tau,~~~~j \in \left\{ \alpha, \beta \right\} .
	\end{split}
	\label{eq:hat_d_ab}
\end{equation}
The nonlinear function $\mu_e(\tilde{e}_j)$ is %%given by 
\begin{equation}
\mu_e(\tilde{e}_j)=l_{p,e} \frac{\tilde{e}_j }{ (\tilde{e}_j/\tilde{e}^{\max})^2+1 }
\label{eq:mu_e}
\end{equation}
where $l_{p,e}$ is the estimation gain related to the peak value of the $\mu_e(\tilde{e}_j)$ as depicted in~Fig.~\ref{fig:dot_hat_tau_L}.
The derivative of $\hat{d}_j$ in (\ref{eq:hat_d_ab}) with respective to time is represented as follow:
\begin{equation}
	\begin{split} 
		\dot{\hat{d}}_j  
		&=        -\frac{d\mu_e(\tilde{e}_j)}{d\tilde{e}_j} \frac{d\tilde{e}_j}{dt} -l_{i,e} \tilde{e}_j.
	\end{split}
\label{eq:dot_hat_d_ab}
\end{equation}
By defining $\partial \mu_e(\tilde{e}_j) :=\frac{d\mu_e(\tilde{e}_j)}{d\tilde{e}_j}$,   from~(\ref{eq:ded2}),  (\ref{eq:hat_d_ab}) and (\ref{eq:dot_hat_d_ab}), the estimation error dynamics of $\tilde{e}_j$ and $\tilde{d}_j$ becomes
\begin{equation}
	\begin{split}
		\dot{\tilde{e}}_j   &= -\frac{R}{L}\tilde{e}_j    -\frac{1}{{J}} \tilde{d}_{j} \\
		\dot{\tilde{d}}_j &=  \left\{ -\partial  \mu_e(\tilde{e}_j) \frac{R}{L}+l_{i,e}\right\}\tilde{e}_j  - \frac{\partial  \mu_e(\tilde{e}_j)}{L}\tilde{d}_j   + \dot{d}_j.
	\end{split}	
\label{eq:DOB_CurrentsDynamics}
\end{equation}   
For simplicity of notation, let us define the estimation error states for inner-loop control system as  
$	\mathbf{d}_{j}=\begin{bmatrix} \tilde{e}_j & \tilde{d}_j\end{bmatrix}^T$, and (\ref{eq:DOB_CurrentsDynamics}) can be represented as 
\begin{equation}
	\label{eq:chi_dynamics}
	\dot{\mathbf{d}}_j  = A_{e}(\tilde{e}_j)\mathbf{d}_j  +B_e \dot{d}_j  \\
\end{equation}
where $A_{e}(\tilde{\omega})  
=   A_{e0} +\partial\mu_e(\tilde{e}_j) {A_{e1} }$, $B_e = \begin{bmatrix} 0 & 1\end{bmatrix}^T$
$A_{e0}=  
{\begin{bmatrix}
		-\frac{R}{L} &   -\frac{1}{L} \\
		l_{i,e} &  0 \\
\end{bmatrix}}$, 
and $A_{e1}={\begin{bmatrix}
		0 & 0\\
		-\frac{R}{L} &  - \frac{1}{L} \\
\end{bmatrix}}$. 
	The value of $\dot{d}_j$ is related with the variance of $\dot{i}_j^d$ and $e_{vj}$. So, in practice, without loss of generality we can assume that there exists $\delta_j$ such as
	\begin{equation}
		%{\delta}^{max}_i = \sup_{\dot{d}_i \in D_i}{ \norm{\dot{d}_i}_{\infty}}.
		{\delta}_j = \sup_{t \in [0, \infty )} { \norm{\dot{d}_j}}.
	\end{equation}
	%
%%%%%%%%%%%%%%%%%%%%%%%%%%%%%%%%%%%%%%%%%%%%%%%%%

	\begin{theorem}{\bf(Exponential convergence of $ {\mathbf{d}_j}$ to bounded ball $B_\varepsilon$ in finite time $T_f$):}
		\label{th:th_currentDOB}
		%
%Suppose  there exists $\delta_j$. Furthermore,		suppose that the estimation error dynamics of inner-loop disturbances are given by~(\ref{eq:chi_dynamics}).  Then,  the estimation law~(\ref{eq:hat_d_ab}) guarantees that the estimation errors for inner-loop disturbances, $\mathbf{d}_j$,  converge exponentially to  the bounded ball $B_\varepsilon $   within finite time $T_f$.
Suppose  there exists $\delta_j$. Further, suppose that there exist $e_j$ and $u_j$.  Then, the estimation law~(\ref{eq:hat_d_ab}) guarantees that the estimation errors for inner-loop disturbances, $\mathbf{d}_j$,  converge exponentially to  the bounded ball $B_\varepsilon $   within finite time $T_f$.
	\end{theorem}

	\begin{proof}
		Let us define the Lyapunov function candidate $V_j(\mathbf{d}_j)$ as
		\begin{equation}
			\begin{split}
				\label{eq:V_chi}	
				V_j(\mathbf{d}_j)& =    \mathbf{d}_j^T P_j \mathbf{d}_j,~~j\in\left\{\alpha, \beta\right\}
			\end{split}
		\end{equation}
			where $P_j=P_j^T >0$.
		The time derivative of $V_j$ becomes
		\begin{equation}
			\begin{split}
				&\dot{V}_j(\mathbf{d}_j) =  \mathbf{d}_j^T( A_{e0}^T P_j   +  P_j A_{e0} ) \mathbf{d}_j \\ &~~~~~~~ +    {\partial \mu}(\tilde{e}_j) \mathbf{d}_j^T(  A_{e1}^T P_j + P_j A_{e1}  )\mathbf{d}_j+2 \mathbf{d}_j^T  P_jB_e {\dot{d}_{j}}.\\
			\end{split}
			\label{eq:Vdot_chi1}
		\end{equation}
		Since we can select $l_{i,e}>0$ satisfying $\sigma(A_{e0}) \subset {\mathbb{C}}^o_{-}$, for any $Q_{e0}=Q_{e0}^T>0$, there exists $P_j=P_j^T >0$ such that~\cite{khalil2002nonlinear}
		\begin{equation}
			\begin{split}
				A_{e0}^T P_j +  P_jA_{e0} = - Q_{e0} .
			\end{split}
		\end{equation} 
		In addition, let us define $Q_{e1}:=A_{e1}^TP_j   +  P_jA_{e1}$. Here, we use the two norm of vector $x$ as $\norm{x}$ and the induced matrix norm of matrix $X$ as $\norm{X}$.
		Then, we see that (\ref{eq:Vdot_chi1}) becomes
		\[
		\dot{V}_j(\mathbf{d}_j)
		= -\mathbf{d}_j^T Q_{e0} \mathbf{d}_j + {\partial \mu}(\tilde{e}_j)   \mathbf{d}_j^T Q_{e1}    \mathbf{d}_j +2 \mathbf{d}_j^T  P_j B_e {\dot{d}_{j}}.
		\]
		Let us define the  minimum/maximum eigenvalues of matrix $X$ as $\lambda_{\min}(X)$, $\lambda_{\max}(X)$.
		Since $Q_{e1}$ can be indefinite and $\norm{B_e} = 1$, there exists $\varepsilon>0$ such that
		\begin{equation}
			{\small
				\begin{split}
					&\dot{V}_j (\mathbf{d}_j)
					\le\\ &-{\lambda_{\min}(Q_{e0})}\norm{\mathbf{d}_j}^2
					+      \abs{ {\partial \mu}(\tilde{e}_j) } 	\norm{Q_{e1}} \norm{\mathbf{d}_j}^2
					+ 2\norm{P_j}\norm{\mathbf{d}_j}   {\delta}_j\\
					&=\\
					&-\left\{\frac{\lambda_{\min}(Q_{e0})} {\lambda_{\max}(P_j)} -    	\abs{ {\partial \mu}(\tilde{e}_j) }     \frac{\norm{Q_{e1}}} {\lambda_{\min}(P_j)}    \right\}V_j(\mathbf{d}_j)  + 2\norm{P_j}\norm{\mathbf{d}_j}    {\delta}_j\\
					& ~~~~~~~~~~~~~~~~~~~~~~~~+ 2 \norm{P_j}\norm{\mathbf{d}_j}^2 \varepsilon_{}^{-1}     	{\delta}_j - 2\norm{P_j} \norm{\mathbf{d}_j}^2\varepsilon_{}^{-1}      	{\delta}_j  \\
					&  \leq - \gamma_e(  \tilde{e}_j   ) {V}_b(\mathbf{d}_j) ,~~~~~~~~~~~~~~~~\forall \norm{\mathbf{d}_j} \ge \varepsilon\\
				\end{split}
				\label{eq:Vdot_chi2}  }	
		\end{equation}
		where,
		\begin{equation}
			\small
			\begin{split}
				&\gamma_e   (\cdot ) =          \frac{\lambda_{\min}(Q_{e0})} {\lambda_{\max}(P_j)} %\\
				%&~~~~~~~~~~~~~~~~
				-    	\abs{ \partial\mu(\tilde{e}_j) }     \frac{\norm{Q_{e1}}} {\lambda_{\min}(P_j)}
				%\\     &~~~~~~~~~~~~~~~~~~~~~~~~~~~~~~~~~~~~~~~~~~~~~~~
				- {2 \frac{\norm{P_j}\varepsilon^{-1}} {\lambda_{\min}(P_j)}     }   {\delta}_j.
			\end{split}
			\label{eq:gamma}
		\end{equation}
		% 
%%%%%%%%%%%%%%%%%%%%%%%%%%%
%
%Although $\gamma_e(\cdot)$ varies according to $\tilde{e}_j$, it is straightforward that   $\abs{ \partial\mu(\tilde{e}_j) }   $    has upper bounded values such as
Although $\gamma_e(  \tilde{e}_j  )$ varies  due to the $\abs{ \partial\mu(\tilde{e}_j) } $, it is straightforward that   $\abs{ \partial\mu(\tilde{e}_j) }   $    has upper bounded values from~(\ref{eq:mu_e}) and~Fig.~\ref{fig:partial_Mu} such as
\begin{equation*}
	%\abs{ \partial\mu(\tilde{e}_j) }   
	\mu^{\max}_{\tilde{e}}=\sup_{\tilde{e}_j \in D_{\tilde{e}}  }{  \abs{ \partial\mu(\tilde{e}_j) }    } = l_{p,e}    
\end{equation*} 
		where $D_{\tilde{e}}=\{\tilde{e}_j : \abs{\tilde{e}_j} \le \tilde{e}^{\max} \} $ and $\tilde{e}_{\max}>0$.
Therefore, given  $\varepsilon >0 $, $Q_{e0}$,  $l_{p,e}$ and $l_{i,e}$  satisfying $\sigma(A_{e0}) \subset {\mathbb{C}}^o_{-}$,    there exists a constant $\gamma_{e}^*  $ such that $\dot{V}_j({\mathbf{d}_j})\le  - \gamma_{e}^*   {V}_j(\mathbf{d}_j), \forall \norm{\mathbf{d}_j} \ge \varepsilon$
such as
\[\gamma_{e}^{*}  = \inf_{\tilde{e}_j \in D_{\tilde{e}}}{  \gamma_e(\tilde{e}_j)} \]
%
%%%%%%%%%%%%%%%%%%%%%%%%%%
 %
 %
% Given $l_{p,e}, l_{i,e}$  and $\varepsilon$, it is straightforward to find $\gamma_e^* >0$ from~(\ref{eq:gamma}).
We need to show the existence of tuple  $(l_{p,e}, l_{i,e})$ ensuring $\gamma_e^*>0$, and  a numerical example will be presented in Sec.~\ref{sec:section4}.
Next, we will show that $\mathbf{d}_j$ converges to the bounded ball, $B_\varepsilon = \{ \mathbf{d}_j|\norm {\mathbf{d}_j}  < \varepsilon \}$ within finite time $T_f$.
With the positive infimum value $\gamma_e^*$, using the Gronwall-Bellman Inequality~\cite{khalil2002nonlinear}, we can get the inequality of $V_{\mathbf{d}}(T_f)$ as
\begin{equation}
 	\begin{split}
 		&V_\mathbf{d}(T_f) =\norm{\mathbf{d}_j(T_f)}^2 = \varepsilon^2 \le V_\mathbf{d}(t_0) e^{-\gamma_e ^* (T_f-T_0)}\\
 	\end{split}
 	\label{eq:Tf_deriv_1}
\end{equation}
where
\[
T_f \le T_0+ \frac{ \log_{e }{V_\mathbf{d}(T_0)}-2\log_{e }{ \varepsilon_{}   }   } {\gamma_e^*}
\] 
Thus, each $\mathbf{d}_j$ gets into $B_\varepsilon$  within finite time $T_f$.
\end{proof}
	%%%%%%%%%%%%%%%%%%%%%%%%%%%%%%%%%%%%%%%%%%%%%%%%%%%%

	\subsection{Lyapunov-based Nonlinear Currents Control }
	%%%%%%%%%%%%
	In this subsection, we present a nonlinear currents control law combined with N-PI-DOB for inner-loop control system. To make currents tracking errors, $\mathbf{e}_{\alpha\beta}$ in (\ref{eq:Errors}), stay within the bounded ball, let us define Lyapunov function candidate for the currents tracking as follow:
	\begin{equation}
		\begin{split}
			\label{eq:V_eab}
		{ V (\mathbf{e}_{\alpha\beta}) = \mathbf{e}_{\alpha\beta}^T \mathbf{e}_{\alpha\beta}} .
		\end{split}
	\end{equation} 
	Then, we can design a nonlinear inputs $\mathbf{u}_{\alpha\beta}$ such as
	\begin{equation}
		\begin{split}
			\mathbf{u}_{\alpha\beta} &=  \begin{bmatrix}
				   -\eta_1 \frac{e_\alpha}{V_{\alpha\beta}+\eta_2} + \hat{d}_\alpha&
				 -\eta_1 \frac{e_\beta}{V_{\alpha\beta}+\eta_2 }+   \hat{d}_\beta
			\end{bmatrix}^T
		\end{split}
		\label{eq:u_lyap}
	\end{equation}
where $\eta_1$ and $\eta_2$ are control gains and positive.
	{

	\begin{theorem}{\bf(Uniform Convergence of Currents Tracking Errors in Finite Time , $T_f$):}
		%		(:)
		%{\it Theorem 1 (\cite{jeong2020nonlinear}):}
		\label{th:th2}
		Suppose that the tracking error dynamics of SPMSM is given by~(\ref{eq:ded_final}).
		With $\eta_1>0$ and $\eta_2 > 0$, $ {V} ({\mathbf{e}_{\alpha\beta}})$, and $\mathbf{u}_{\alpha\beta}$,
		the currents tracking errors, $\mathbf{e}_{\alpha\beta}$, converge into  bounded ball.
	\end{theorem}
	
	\begin{proof}
		\label{prov:prov1}
		The derivative of Lyapunov function candidate with respect to time, $t$, becomes
%\dot{\mathbf{e}}_{\alpha\beta}  &=  A_{\alpha\beta} {\mathbf{e}}_{\alpha\beta} + B_{\alpha\beta} ({\mathbf{u}}_{\alpha\beta}  -  {\mathbf{d}}_{\alpha\beta})
	\begin{align}
		\begin{split}
			\label{eq:sgm_dot_new2}
			\frac{d}{dt}  {V} ({\mathbf{e}_{\alpha\beta}})&=
			{\mathbf{e}_{\alpha\beta}}^T{\dot{\mathbf{e}}_{\alpha\beta}}+{\dot{\mathbf{e}}_{\alpha\beta}}^T{\mathbf{e}_{\alpha\beta}}\\
			&=
		    2{\mathbf{e}_{\alpha\beta}}^TA_{\alpha\beta} 	{\mathbf{e}}_{\alpha\beta} + 2{\mathbf{e}_{\alpha\beta}}^TB_{\alpha\beta} ({\mathbf{u}}_{\alpha\beta}  -  {\mathbf{d}}_{\alpha\beta})      \\ 
		    &=
		    -2\frac{R}{L}{\mathbf{e}_{\alpha\beta}}^T{\mathbf{e}_{\alpha\beta}} - 
		    \frac{2}{L}{\mathbf{e}_{\alpha\beta}}^T ({\mathbf{u}}_{\alpha\beta}  	-  {\mathbf{d}}_{\alpha\beta}).\\    
		\end{split}
	\end{align}
		By plugging (\ref{eq:u_lyap}) into (\ref{eq:sgm_dot_new2}), then, we see that (\ref{eq:sgm_dot_new2}) becomes
		
		\begin{align}
			\begin{split}
				\label{eq:sgm_dot_new3}
				\frac{d}{dt}  {V} ({\mathbf{e}_{\alpha\beta}})  
				&\le -\frac{2\eta_1}{L} \frac{V(\mathbf{e}_{\alpha\beta})}{V(\mathbf{e}_{\alpha\beta})+\eta_2}
				+ \frac{2}{L} \abs{e_\alpha\tilde{d}_\alpha + e_\beta \tilde{d}_\beta} 
			\end{split}
		\end{align}
		From  {\bf{Theorem}~\ref{th:th_currentDOB}}, we prove that $\tilde{d}_\alpha, \tilde{d}_\beta $ converge to bounded ball within finite time $T_f$.	
		Therefore, given  $ {e}^{\max}>0   $,  there exists
		\begin{equation}
			{\delta}_{\alpha\beta} = \sup_{ \mathbf{e}_{\alpha\beta} \subset D_e  }{  \frac{2}{L} \abs{e_\alpha\tilde{d}_\alpha + e_\beta \tilde{d}_\beta}} 
		\end{equation}
		where $D_{{e}}=\{\mathbf{e}_{\alpha\beta} : \norm{\mathbf{e}_{\alpha\beta}} \le {e}^{\max} \} $.
	 Let $f(x) = \frac{x}{ x+\eta_2}$. As shown in Fig.~1 of \cite{jeong2020nonlinear},  given $\eta_2>0$ we see that  $\lim_ {x  \rightarrow \infty } {f(x)} = 1$. Thus   given $\eta_2>0$ and  for $x \geq \varepsilon_{\alpha\beta}$, we see that there exists the infimum of $f(x)$ for $x \geq \varepsilon_{\alpha\beta}$.
		With  $ \Xi = \{ x~ | ~f(x ) \geq \frac{\varepsilon_{\alpha\beta}}{\varepsilon_{\alpha\beta}+\eta_2}, ~~\forall ~ x \}$ , let $\rho (\eta_2) = \inf_{ x \in \Xi} \frac{x}{x + \eta_2  }$. As a result, we have
		\[
		\dot V(\mathbf{e}_{\alpha\beta}) \leq - \frac{{2\eta_1 }}{L} \rho(\eta_2)+ \delta_{\alpha\beta}.
		\]
		Therefore, it is clear that the control law~(\ref{eq:u_lyap}) ensures 
		currents tracking errors converge to bounded ball.
		
		%We see that \inf_{  (e_\alpha , e_\beta ) \in \mathbb{R}^2 } $\frac{V}{{V + \varepsilon } = \gamma /(\gamma + \epsilon)$.
		%
		%After that, The phase current tracking errors are bounded such as $\left| {{e_\alpha }} \right| \le \sqrt {2n\epsilon}  ,\,\,\left| {{e_\beta }} \right| \le \sqrt {2n \epsilon}$ since the $2V$ is the sum of square of each current tracking errors.
		%
		
		%	\begin{figure}[!b]
		%		\centering
		%		{\includegraphics[width=0.499\textwidth]{Fig_NewC/fx.eps}}
		%		\caption{Graph of $f(x)=\frac{x}{x+\epsilon}$}
		%		\label{fig:fx}
		%	\end{figure}
		%
	\end{proof}
	
	%
%	Until now, we have discussed that the stability analysis of Lyapunov-based currents controller with N-PI-DOB for inner-loop control system, and prove
	
Until now, we have proved that $\mathbf{e}_{\alpha\beta}$ converge to bounded ball with Lyapunov-based currents controller and N-PI-DOB for inner-loop control system.
In next subsection, we will discuss the stability of outer-loop control system.

	\subsection{Stability of Outer-Loop Control System}
	\label{Sec:cascade_stability} 
	In this subsection, we analyse stability of the N-PI-DOB for outer-loop control system after the currents tracking errors converged to bounded ball.
	After then, we will show that the motion tracking system has Intput-to-State Stability~(ISS) properties~\cite{khalil2002nonlinear}.

	\begin{theorem}{\bf(Exponential Convergence of inner-loop Disturbance Estimation Error):}
		\label{th:th3}
		%

%	Suppose  there exists $\delta_\tau$ in~(\ref{eq:delta_tau}). With the torque modulation and N-PI-DOB for outer-loop control system~(\ref{eq:tau_hat}), the estimation error dynamics of outer-loop disturbances are given by~(\ref{eq:dtau_dynamics}). Then, the estimation error of unknown load torque, $\mathbf{d}_\tau$,  starts to converge exponentially to bounded ball $B_{\varepsilon_\tau} $  once $\mathbf{e}_{\alpha\beta}$ converges to bounded ball.

	Suppose  there exists $\delta_\tau$ in~(\ref{eq:delta_tau}). Given states of SPMSM, torque modulation and N-PI-DOB for outer-loop control system~(\ref{eq:tau_hat}), the estimation error dynamics of outer-loop disturbances are given by~(\ref{eq:dtau_dynamics}). Then, the estimation error of unknown load torque, $\mathbf{d}_\tau$,  starts to converge exponentially to bounded ball $B_{\varepsilon_\tau} $  once $\mathbf{e}_{\alpha\beta}$ converges to bounded ball.

	\end{theorem}
	%%%
	%
	%

	\begin{proof}
		Let us define Lyapunov function candidate $V(\mathbf{d}_{\tau})$ as
		\begin{equation}
			\begin{split}
				\label{eq:V_tau}	
				V(\mathbf{d}_{\tau})& =    \mathbf{d}_{\tau}^T P_{\tau} \mathbf{d}_{\tau} 
			\end{split}
		\end{equation}
		where $P_{\tau}=P_{\tau}^T >0$.
		The time derivative of $V$ becomes
		\begin{equation}
			\begin{split}
				&\dot{V}(\mathbf{d}_{\tau}) =  \mathbf{d}_{\tau}^T( A_{\tau0}^T P_{\tau}   +  P_{\tau} A_{\tau0} ) \mathbf{d}_{\tau} \\ &~~~~~~~ +    {\partial \mu}(\tilde{\omega}) \mathbf{d}_{\tau}^T(  A_{\tau1}^T P_{\tau} + P_{\tau} A_{\tau1}  )\mathbf{d}_{\tau}+2 \mathbf{d}_{\tau}^T  P_{\tau}B_\tau {\dot{ {\tau}}_L}.\\
			\end{split}
			\label{eq:Vdot_tau1}
		\end{equation}
		Since we can select $l_{i,\tau}>0$ satisfying $\sigma(A_{\tau0}) \subset {\mathbb{C}}^o_{-}$, for any $Q_{\tau0}=Q_{\tau0}^T>0$, there exists $P_{\tau}=P_{\tau}^T >0$ such that~\cite{khalil2002nonlinear}
		\begin{equation}
			\begin{split}
				A_{\tau0}^T P_\tau +  P_\tau A_{\tau0} = - Q_{\tau0} .
			\end{split}
		\end{equation} 
\begin{figure}[!t]
\centering
{\includegraphics[width=0.370\textwidth]{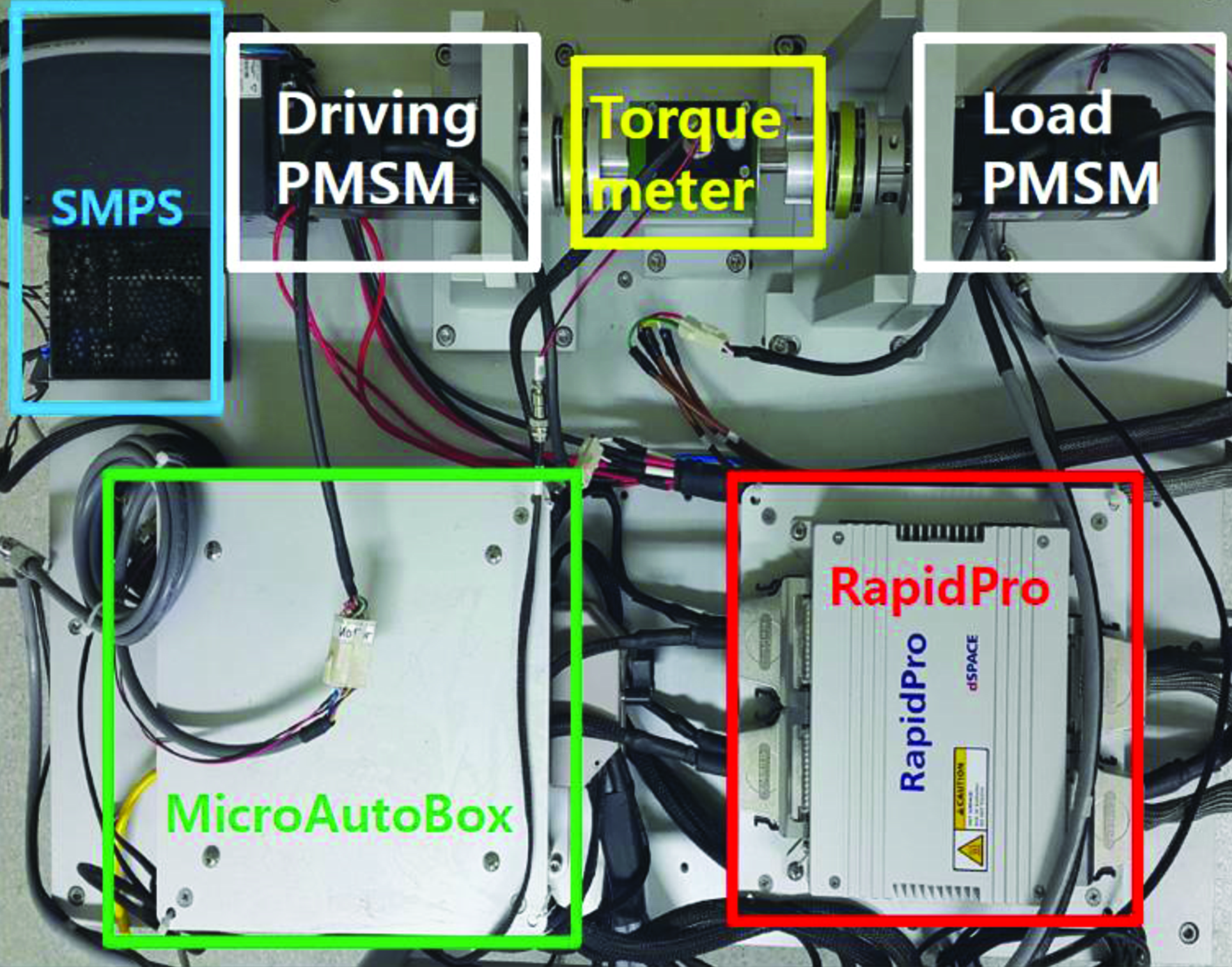}}
%		\caption{Motor Generator Experiment Equipment}
\caption{  Experimental setup: One PMSM (APM-SB03ADK-9, LS Mecapion $\&$ Co) is located between an encoder (2500 pulses per revolution) and a coupler. The torque meter is connected between couplers. Another PMSM locates in series to generate load torque .
	% 	The purpose of a torque sensor was to measure the load torque applied to the driving PMSM by measuring the different angles of two shafts. However, because of the coupler, which interconnect between the torque sensor and the driving PMSM, 	measured data from the torque sensor does not exactly match with load torque applied to the driving PMSM.
	%
}
\label{fig:Exp_Setup}
\end{figure}
		In addition, let us define $Q_{\tau1}:=A_{\tau1}^TP_{\tau}   +  P_{\tau}A_{\tau1}$. Here, we use the two norm of vector $x$ as $\norm{x}$ and the induced matrix norm of matrix $X$ as $\norm{X}$.
		Then, we see that (\ref{eq:Vdot_chi1}) becomes
		\[
		\dot{V}(\mathbf{d}_{\tau})
		= -\mathbf{d}_{\tau}^T Q_{\tau0} \mathbf{d}_{\tau} + {\partial \mu}(\tilde{\omega})   \mathbf{d}_{\tau}^T Q_{\tau1}    \mathbf{d}_{\tau} +2 \mathbf{d}_{\tau}^T  P_{\tau} B_\tau {\dot{ {\tau}}_L}.
		\]
		Let us define the  minimum/maximum eigenvalues of matrix $X$ as $\lambda_{\min}(X)$, $\lambda_{\max}(X)$.
		Since $Q_{\tau1}$ can be indefinite and $\norm{B_\tau} = 1$, there exists $\varepsilon_\tau>0$ such that
		\begin{equation}
			{\small
				\begin{split}
					&\dot{V}  (\mathbf{d}_{\tau})
					\le\\ &-{\lambda_{\min}(Q_{\tau0})}\norm{\mathbf{d}_{\tau}}^2
					+      \abs{ {\partial \mu}(\tilde{\omega}) } 	\norm{Q_{\tau1}} \norm{\mathbf{d}_{\tau}}^2
					+ 2\norm{P_{\tau}}\norm{\mathbf{d}_{\tau}}   {\delta}_\tau\\
					&=\\
					&-\left\{\frac{\lambda_{\min}(Q_{\tau0})} {\lambda_{\max}(P_{\tau})} -    	\abs{ {\partial \mu}(\tilde{\omega}) }     \frac{\norm{Q_{\tau1}}} {\lambda_{\min}(P_{\tau})}    \right\}V(\mathbf{d}_{\tau})  + 2\norm{P_{\tau}}\norm{\mathbf{d}_{\tau}}    {\delta}_\tau\\
					& ~~~~~~~~~~~~~~~~~~~~~~~~+ 2 \norm{P_{\tau}}\norm{\mathbf{d}_{\tau}}^2 \varepsilon_\tau^{-1}     	{\delta}_\tau - 2\norm{P_{\tau}} \norm{\mathbf{d}_{\tau}}^2\varepsilon_\tau^{-1}      	{\delta}_\tau  \\
					&  \leq - \gamma_\tau(  \tilde{\omega}   ) {V} (\mathbf{d}_{\tau}) ,~~~~~~~~~~~~~~~~\forall \norm{\mathbf{d}_{\tau}} \ge \varepsilon_\tau\\
				\end{split}
				\label{eq:Vdot_tau2}  }	
		\end{equation}
		where,
		\begin{equation}
			\small
			\begin{split}
				&\gamma_\tau   (\cdot ) =          \frac{\lambda_{\min}(Q_{\tau0})} {\lambda_{\max}(P_{\tau})} %\\
				%&~~~~~~~~~~~~~~~~
				-    	\abs{ \partial\mu(\tilde{\omega}) }     \frac{\norm{Q_{\tau1}}} {\lambda_{\min}(P_{\tau})}
				%\\     &~~~~~~~~~~~~~~~~~~~~~~~~~~~~~~~~~~~~~~~~~~~~~~~
				- {2 \frac{\norm{P_{\tau}}\varepsilon_\tau^{-1}} {\lambda_{\min}(P_{\tau})}     }   {\delta}_\tau.
			\end{split}
			\label{eq:gamma_tau}
		\end{equation}
		% 
		%%%%%%%%%%%%%%%%%%%%%%%%%%%
		%
		%%%%%%%%%%%%%%%%%%%%%%%%%%%%%%%%%%%%%%%%%%%%%%%
		It is straightforward that   $\abs{ \partial\mu(\tilde{\omega}) }   $    has upper bounded values such as
		\begin{equation*}
			%\abs{ \partial\mu(\tilde{e}_j) }   
			\mu^{\max}_{\tilde{\omega}}=\sup_{\tilde{\omega} \in D_{\tilde{\omega}}  }{  \abs{ \partial\mu(\tilde{\omega}) }    } = l_{p,\tau}    
		\end{equation*} 
		where $D_{\tilde{\omega}}=\{\tilde{\omega} : \abs{\tilde{\omega}} \le \tilde{\omega}^{\max} \} $ and $\tilde{\omega}^{\max}>0$.
		Therefore, given  $\varepsilon_\tau >0 $, $Q_{\tau0}$,  $l_{p,\tau}$ and $l_{i,\tau}$  satisfying $\sigma(A_{\tau0}) \subset {\mathbb{C}}^o_{-}$,    there exists a constant $\gamma_{\tau}^*  $ such that $\dot{V} ({\mathbf{d}_\tau})\le  - \gamma_{\tau}^*   {V} ({\mathbf{d}_\tau}), \forall \norm{{\mathbf{d}_\tau}} \ge \varepsilon_\tau$
		such as
		\[\gamma_{\tau}^{*}  = \inf_{\tilde{\omega} \in D_{\tilde{\omega}}}{  \gamma_\tau(\tilde{\omega})} \]
		%
		%%%%%%%%%%%%%%%%%%%%%%%%%%
		%
		%
		% Given $l_{p,\tau}, l_{i,\tau}$  and $\varepsilon_\tau$, it is straightforward to find $\gamma_e^* >0$ from~(\ref{eq:gamma}).
		We need to show the existence of tuple  $(l_{p,\tau}, l_{i,\tau})$ ensuring $\gamma_\tau^*>0$, and  a numerical example will be presented in Sec.~\ref{sec:section4}.
		Next, we will show that $\mathbf{d}_\tau$ converges to the bounded ball, $B_{\varepsilon_\tau} = \{ \mathbf{d}_\tau|\norm {\mathbf{d}_\tau}  < \varepsilon_\tau \}$ within finite time $T_f$.
		With the positive infimum value $\gamma_{\tau}^*$, using the Gronwall-Bellman Inequality~\cite{khalil2002nonlinear}, we can get the inequality of $V_{\mathbf{d}}(T_f)$ as
		\begin{equation}
			\begin{split}
				&V(T_f) =\norm{\mathbf{d}_\tau(T_f)}^2 = \varepsilon_\tau^2 \le V(t_0) e^{-\gamma_\tau^* (T_f-T_0)}\\
			\end{split}
			\label{eq:Tf_deriv_1}
		\end{equation}
		where
		\[
		T_f \le T_0+ \frac{ \log_{e }{V(T_0)}-2\log_{e }{ \varepsilon_\tau    }   } {\gamma_{\tau}^*}
		\] 
		Thus, each $\mathbf{d}_\tau$ gets into $B_{\varepsilon_\tau}$  within finite time $T_f$.

	\end{proof}

\begin{figure} [!t]
	%	\centering
	\raggedleft
	\subfigure[Angular Velocity Reference]{
		\includegraphics[width=0.975\linewidth]{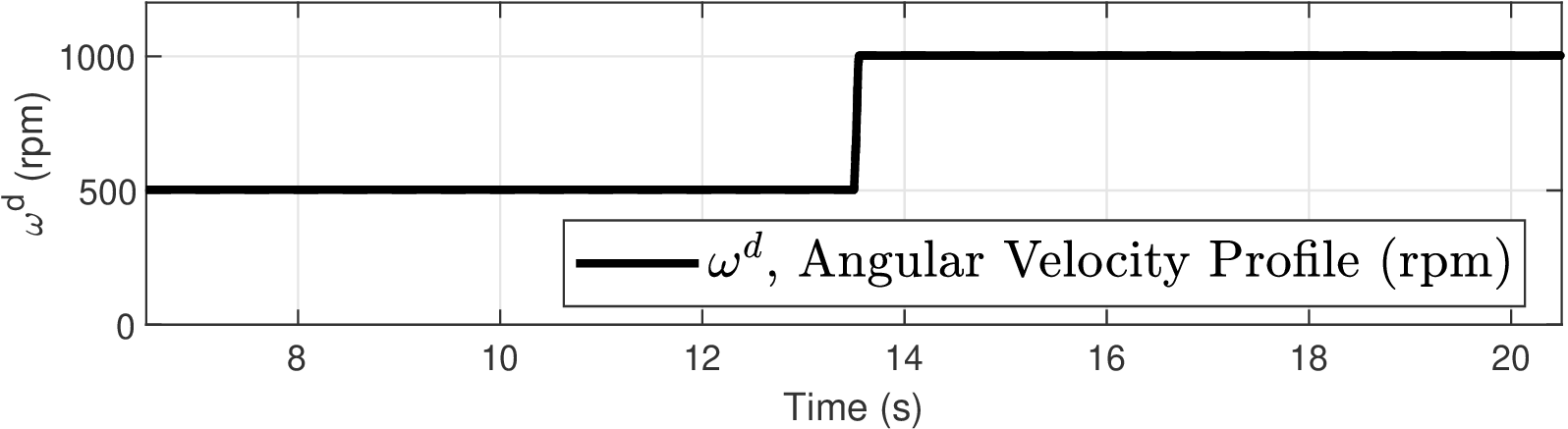}
		\label{fig:vel_ref}
	}
	\subfigure[Measured Load Torque]{
		\includegraphics[width=0.975\linewidth]{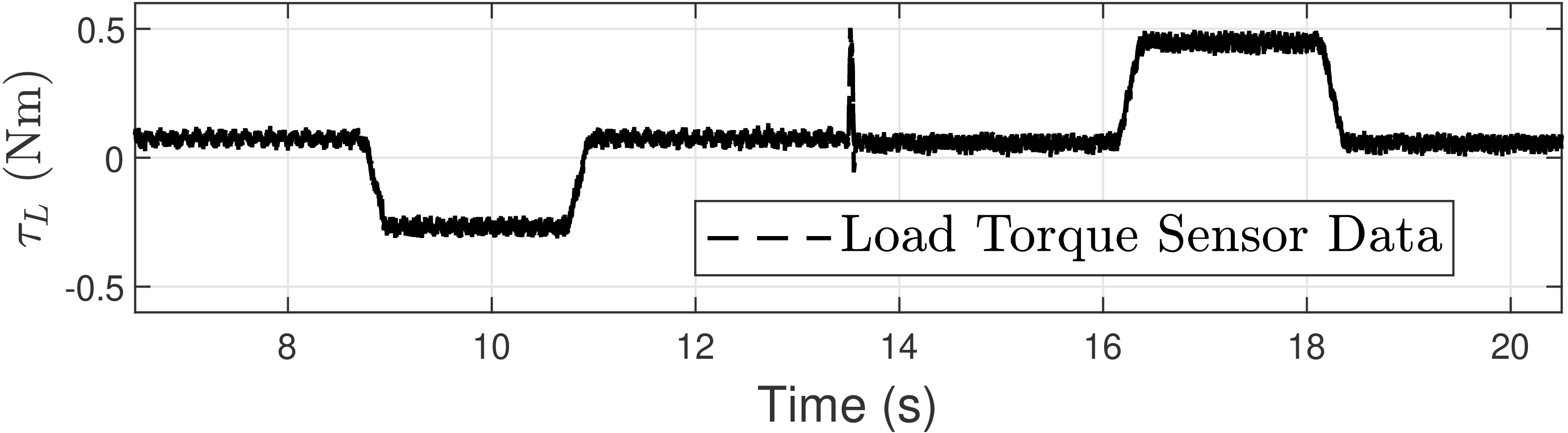}
		\label{fig:Tload_ref}
	}
	\caption{  Angular velocity reference and measured load torque: (a) angular velocity reference and (b)  measured load torque data. Another PMSM connected in series intentionally generated load torque at 8.5 (s) and 16.5 (s), respectively.
	}
	\label{fig:Vel_Tload_reference}
\end{figure}

\begin{table}[!b]
	\renewcommand{\arraystretch}{1.3}
	\caption{Parameters and Control/Estimation Gain}
	\centering
	\footnotesize
	\label{tb:table_1}
	%\centering
	\resizebox{\columnwidth}{!}{
		\begin{tabular}{l l l l l l  }
			\hline\hline \\[-3mm]
			\multicolumn{1}{c}{Symbol}   & \multicolumn{1}{c}{Values}    & \multicolumn{1}{c}{Unit}   &\multicolumn{1}{c}{Symbol}   & \multicolumn{1}{c}{Values}     & \multicolumn{1}{c}{Unit}   \\[1.6ex] \hline
			$J$				 &$4.46\times10^{-4}$& $({\rm kg\cdot m^2})$   &
			${L}$		 &$0.275$ &$\rm (mH)$ \\
			$B$				 &$7\times 10^{-4}$ &$\rm (N\cdot m \cdot s/rad) $ &
			${R}$	 	 &$0.875$&$(\Omega)$\\
			${\Phi}$ 	 & $1.58 \times 10^{-2}$ &$(N\cdot m \cdot A)$ &
			$P$				 &$4$&-\\
			$f_{ctrl}$		 &$100$ &$(us)$  &
			$f_{PWM}$  		 &$20$	&$(kHz)$\\
			%
%			${R}$	 	 &$0.875$&$(\Omega)$&
%			$P$				 &$4$&-\\
			%
			$k_\theta$ 	 		 & $0.4$& -&
			$k_\omega$ 		 &$0.0214$ &-\\
			$l_{p,\tau}$ 			 &$0.001$ & - &
				$l_{p,e}$  	  & $0.0025$&-\\
			$l_{i,\tau}$	           &$10\times10^4$  &-&
			$l_{i,e}$	  &$6\times10^4$ &- \\
			$\tilde{\omega}^{\max}$  & $52.36$    &-&
			$\tilde{e}^{\max}$	  &$15.6$&-\\
			%%%
			$\eta_1$	 	 &$2000$&-&
			$\eta_2$	  	 &$1300$ &- \\
			%
			%
			%$l_{\chi1}$  	  & $1\times10^4$& -&				-    & -& -\\
			%
			%
			%
			\hline\hline
		\end{tabular}
	}
\end{table}

	%%%%%%%%%%%%%%%%%%%%%%%%%%%%%%%%%%%%%%%%%%%%%%%%%%%
	%
	%
	%		Theorem~\ref{th:th2} is obvious since the
	%w
	In  summary, from Theorem~\ref{th:th2}, the currents tracking errors converge to bounded ball. Further, from Theorem~\ref{th:th3}, the load torque estimation error is bounded within $B_{\varepsilon_\tau}$. 
	Since $\sigma\left(A_\eta\right)\subset \mathbb{C}^{o}_{-}$, the dynamics of $\mathbf{e}_{m}$ is an exponentially stable system perturbed by $\mathbf{d}_\tau$ and $\mathbf{e}_{\alpha\beta}$. Again, from the Input-to-State Stability~(ISS) properties~\cite{khalil2002nonlinear} of a tracking error dynamics (\ref{eq:ded_final}), it is obvious that the motion tracking error converges into a small bounded ball.
	%%%%%%%%%%%%%%%%%%%%%%%%%%%%%%%%%%%%%%%%%%%%%%%%%%%%%%%%%%%%%%%%%%%%%%%%%%%%

	\section{Experimental Results}
	\label{sec:section4}

 \begin{figure}[!t]
 	\subfigure[ $\hat{d}_\alpha,\hat{d}_\beta$ ~(Zoom-in, $\omega$ : 500rpm, $\tau_L$ : -0.3 to 0 Nm)]{
 		\includegraphics[width=0.975\linewidth]{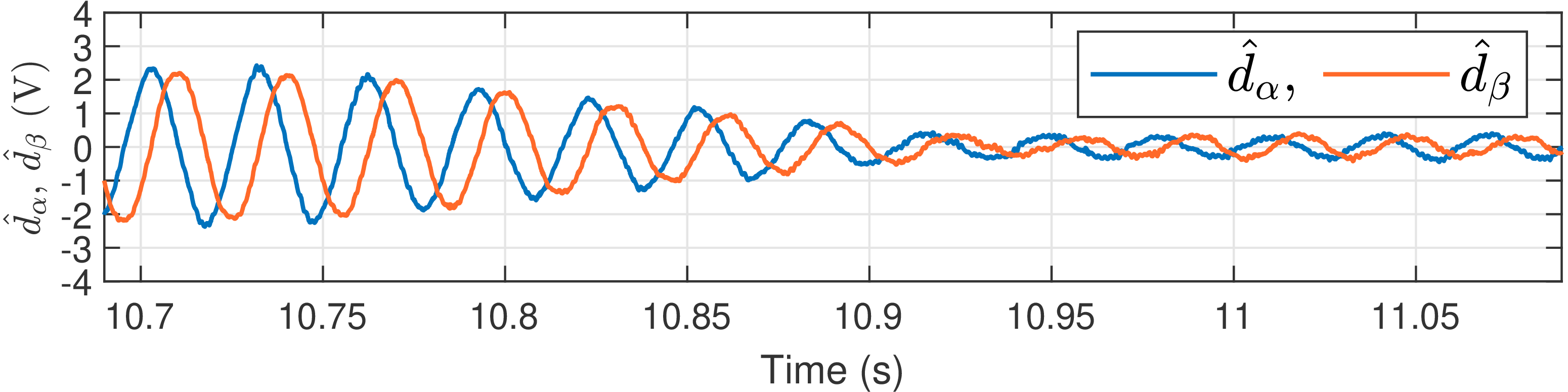}\label{fig:dab_zoom1}}
 	\subfigure[ $\hat{d}_\alpha,\hat{d}_\beta$ ~(Zoom-in, $\omega$ : 1000 rpm, $\tau_L$ : 0 to 0.3 Nm)]{
 		\includegraphics[width=0.975\linewidth]{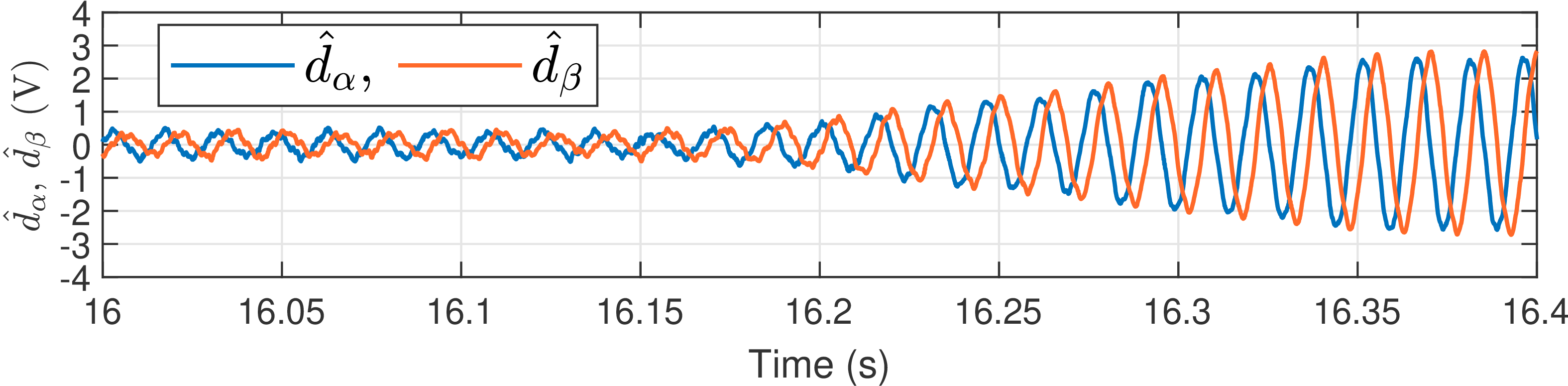}\label{fig:dab_zoom3}
 	}
 	\caption{ Estimated inner-loop disturbances $\hat{d}_\alpha,\hat{d}_\beta$: (a)  Case 1) around 10.9(s), and (b)  Case 2)  around 16.2(s).
 	}
 	\label{fig:dab}
 \end{figure}
 \begin{figure}[!t]
 	\raggedright
 	\subfigure[Tracking performances of $i_\alpha, i_\beta$, and torque reference~(Zoom-in, $\omega$:500rpm, $\tau_L$:-0.3 to 0Nm)]{
 		\includegraphics[width=0.975\linewidth]{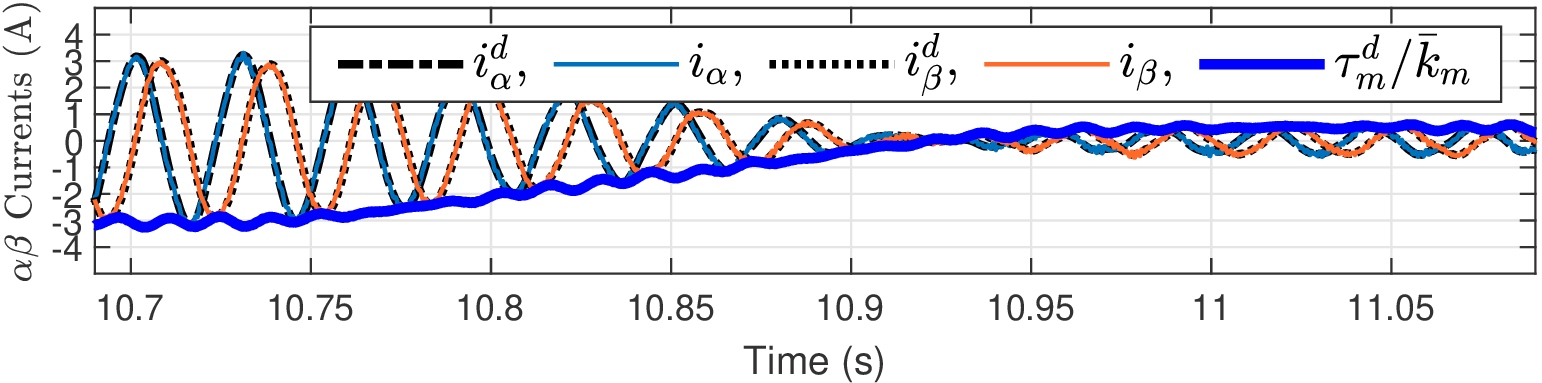}\label{fig:iab_zoom1}
 	}
 	\subfigure[Tracking performances  of $i_\alpha, i_\beta$, and torque reference~(Zoom-in, $\omega$:1000rpm, $\tau_L$:0 to 0.3Nm)]{
 		\includegraphics[width=0.975\linewidth]{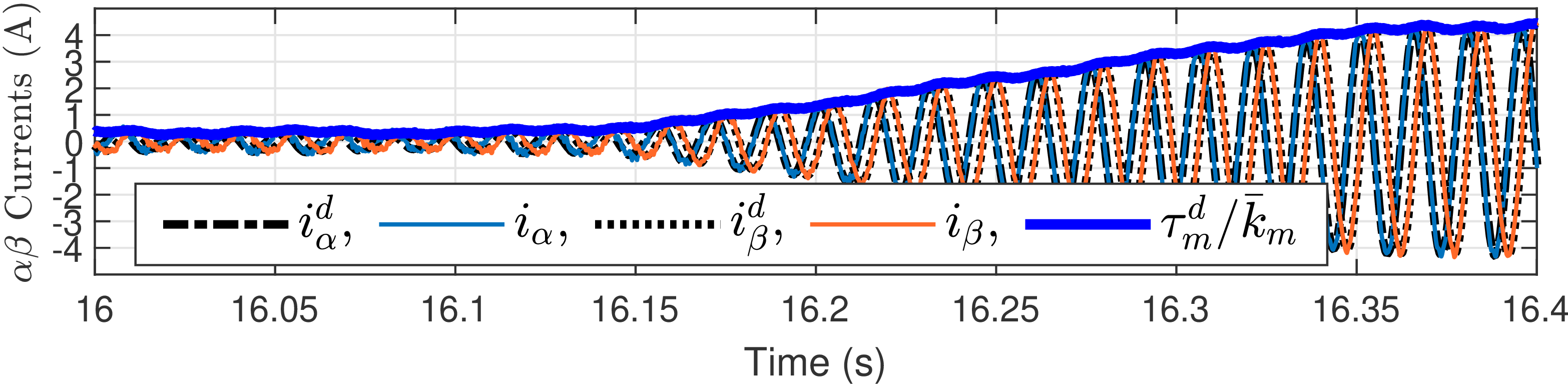}\label{fig:iab_zoom3}
 	}
 	%	\caption{\color{red}Currents Reference(black), Envelope of the Reference Currents(Dark Blue), and Currents Measurements(Red, Sky Blue)}
 	\caption{%
 		$\alpha$/$\beta$-phase currents references and measured ones, and torque reference : (a)  Case 1) around 10.9(s), and (b)  Case 2)  around 16.2(s).
 	}
 	\label{fig:iab}
 \end{figure}

	This section describes experimental  results for validating the proposed outer/inner-loop N-PI-DOB and Lyapunov-based nonlinear currents controller for precision motion control of SPMSM as illustrated in~Fig.~\ref{fig:N_PI_DOB_Control}.
	Figure~\ref{fig:Exp_Setup}. shows the motor and generator set~(MG Set).
	%
	% Lee Paper LPV.
	Control logics for the outer/inner-loop system were implemented in MicroAutobox embedded computer~(dSPACE). Furthermore, the 3-phase motor driver system was implemented by RapidPros~\cite{lee2017lpv}.
	Each RapidPro unit includes three half-bridge power stage modules, and the two switches of each half-bridge were driven by complementary signals with some dead time to avoid feedthrough fault.
	The sample rate of control logics,~$f_{ctrl}$, was 10kHz. The switching frequency of pulse width modulation,~$f_{pwm}$, was 20 kHz.
	The parameters of SPMSM and control/estimation gains are listed in Table.~\ref{tb:table_1}.
	Given PMSM motor parameters, and estimation gains and $\delta_{\tau}=\delta_{j}=1$, $\varepsilon =\varepsilon_\tau = 0.1$, $Q_{\tau0}=Q_{e0}=1000\times I_{2\times2}$,
	 we see that $\gamma_{\tau}^*  = 0.03>0$  and $\gamma_{e}^*  = 50.4>0$.
	%%%%%%%%%%%%%%%%%%%%%%%%%%%%%%%%%%%%%%%%
	%
	As shown in  Fig.~\ref{fig:Exp_Setup}, one SPMSM (APM-SB03ADK-9, LS Mecapion $\&$ Co) was located between an encoder (2500 pulses per revolution) and a coupler. Then, the torque meter is connected between couplers and another PMSM for generating load torque locate in series.
	Motor speed was obtained by  time stamped method using dSPACE AC Motor Control Solutions.
	%
	% 	The purpose of a torque sensor was to measure the load torque applied to the driving PMSM by measuring the different angles of two shafts. However, because of the coupler, which interconnect between the torque sensor and the driving PMSM, 	measured data from the torque sensor does not exactly match with load torque applied to the driving PMSM.
	%
	Figure~\ref{fig:Vel_Tload_reference} shows the velocity reference and the measured load torque.
The desired velocity reference was changed from 500 rpm
to 1,000 rpm around 13.5 sec.	To validate the robustness for the outer/inner-loop N-PI-DOB performances, as illustrated in Fig.~\ref{fig:Tload_ref}, we intentionally injected the load torque.

	%%%%%%%%%%%%%%%%%%%%%%%%

\begin{figure}[t]
	\centering
	\subfigure[  $\norm{\mathbf{e}_{\alpha\beta}}_2$]{
		\includegraphics[width=0.93\linewidth]{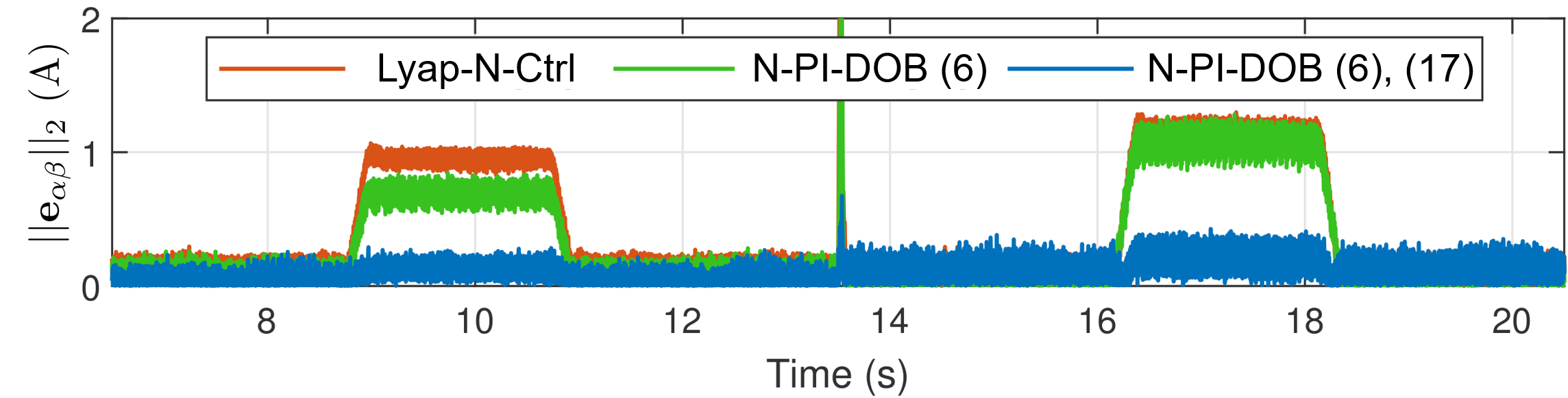}
		\label{fig:xi}
	}
	\subfigure[  $\tilde{\tau}_L$]{
		\includegraphics[width=0.95\linewidth,left]{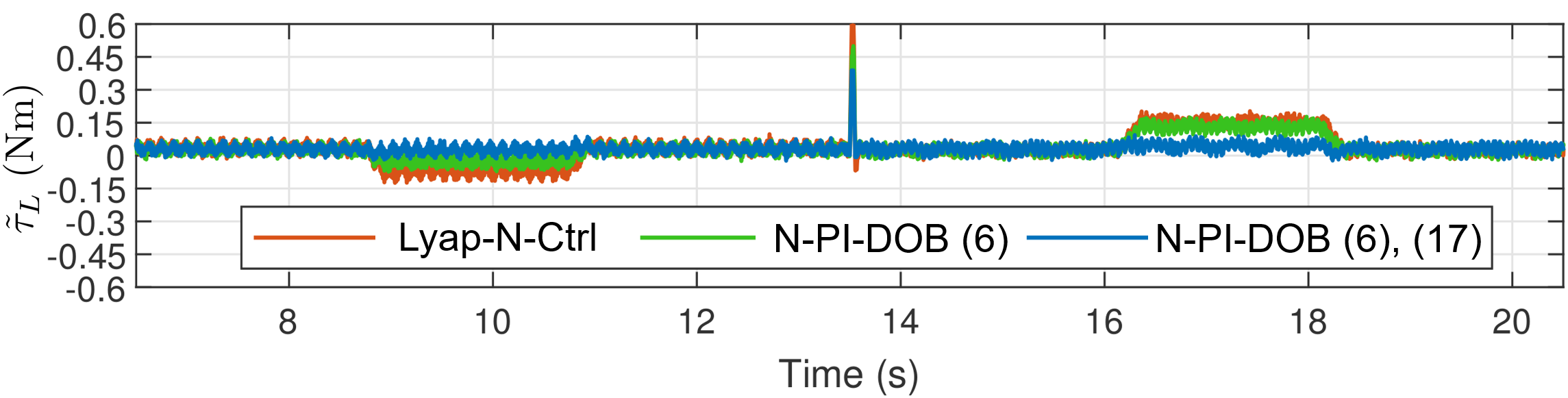}\label{fig:zeta2}
	}
	%%%%%%%%%%%%%%%%%%
	\subfigure[  {$e_\omega$} ]{
		\includegraphics[width=0.95\linewidth,left]{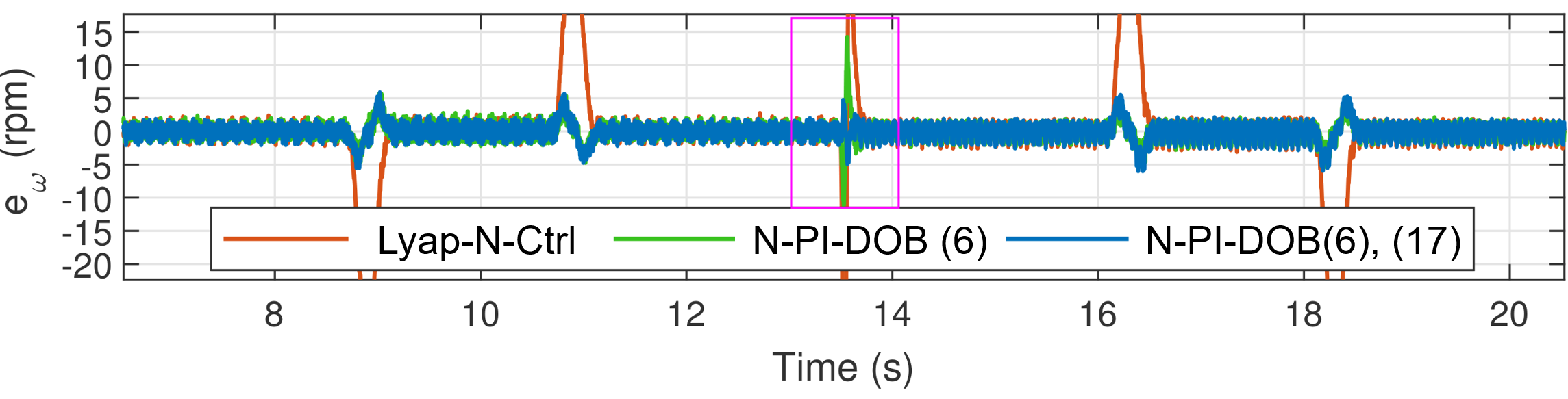} \label{fig:vel_comp_Exp}
	}
	\subfigure[  { $e_\omega$ Zoom-In  around 13.5 (s)} ]{
		\includegraphics[width=0.95\linewidth,left]{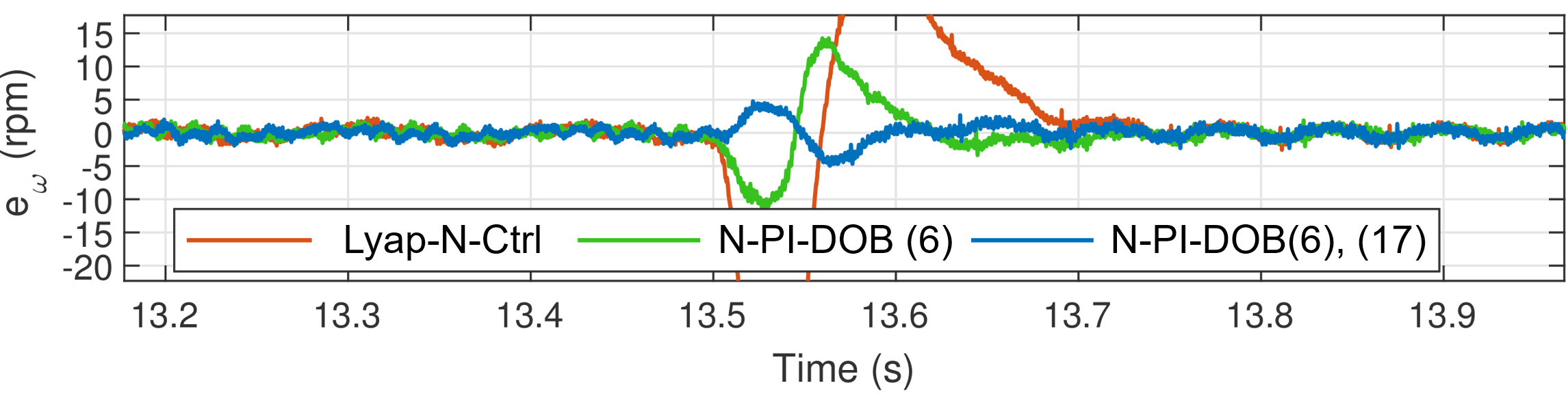} \label{fig:vel_comp_Exp_zoom}
	}
	\caption{ Comparative study of currents tracking, load torque estimation, and motion tracking Performance~($ {\mathbf{e}_{\alpha\beta}}_2, \tilde{\tau}_L, e_\omega$) of PMSMs: (a) currents tracking performances, $\norm{\mathbf{e}_{\alpha\beta}}_2$, (b) load torque estimation error, $\tilde{\tau}_L$, (c) velocity tracking error, $e_\omega$, and (d)  velocity tracking error, $e_\omega$ zoomed around 13.5(s) when the velocity reference was changed. }
	\label{fig:Cascade_Performance}
\end{figure}
%%

%% Not decided yet	
%	With these validation environments, we validate the effectiveness of the inner-loop disturbance estimation/compensation logics.
	%
	Figure~\ref{fig:dab} shows the estimation results of inner-loop disturbances~($\hat{d}_\alpha, \hat{d}_\beta$) around 10.9 sec and 16.2 sec, respectively.
 Figure.~\ref{fig:iab} shows the currents tracking performances~($i_\alpha^d$ vs $i_\alpha$, $i_\beta^d$ vs $i_\beta$).
Notice that there is no way to measure the exact values of ${d}_\alpha, d_\beta$. However, we can confirm its performances indirectly by observing current tracking performances as illustrated in Fig.~\ref{fig:iab}.
   In Fig.~\ref{fig:iab}, the red/blue lines indicate the graph of $\alpha$/$\beta$-phase signals. We present experimental results  in two cases: Case 1) 500 rpm  speed region with the load torque was reset from -0.3 to 0 Nm  as shown in~Fig.~\ref{fig:dab_zoom1}, Case 2) 1000 rpm  speed region with the load torque was applied from 0 to 0.3 Nm as shown in~Fig.~\ref{fig:dab_zoom3}.
	Here, we observed that each of estimated inner-loop disturbances has $33.33$ Hz and $66.66$ Hz  frequency components, identically same as with the electric angular speed.
	As we expect that disturbances come from the parameter uncertainties and the phase delay of the low-pass filter, the amplitude of disturbances changed as the amplitude of  $i_\alpha^d, i_\beta^d, i_\alpha, i_\beta$ varied.
%
	%%%%%%%%%%%%%%%%%%%%%%%%%%%%
	
	%
	%
	Figure~\ref{fig:Cascade_Performance} shows the comparative studies for analyzing tracking/estimation performances~($\mathbf{e}_{\alpha\beta}, \mathbf{d}_{\tau},\mathbf{e}_{m}$) between the methods in~\cite{jeong2020nonlinear}  and N-PI-DOB with Lyapunov-based nonlinear currents controller.
	%%
	%
	% Modification not applied
	Figure~\ref{fig:xi} shows the currents tracking performances.
	We expect that the inner-loop N-PI-DOB and Lyapunov-based nonlinear currents controller enhance the currents tracking performances. And we observed that the amplitude of~$\norm{\mathbf{e}_{\alpha\beta}}_2$ of the proposed method uniformly remains within the bounded ball although the velocity reference changed and load torque was injected. However, the other previous method and without N-PI-DOB for inner-loop control system  do not show the uniform currents tracking performances in same conditions.
	Figure~\ref{fig:zeta2} shows a comparative study of load torque estimation performance, $\tilde{\tau}_L$.
	%		
	%\textcolor[rgb]{1.00,0.00,0.00}{We exepct that the uniform currents tracking performance guarantees the uniform load torque estimation performances. } We observed that $\zeta_2$ uniformly remains within the bounded ball. However, the other methods do not show the uniform convergence in same conditions.
	%
	%We observed that the load torque estimation error, $\zeta_2$, stays within $0.1$(Nm) except the time region, $13.5\sim13.6$(s), when the velocity change from 500 to 1000(rpm).
	%We expect that the abnormal peak of load torque measurements in $13.55$(s) was due to the coupler installed between the driving PMSM and the torque sensor.
	%With these validation results, we show the effectiveness of the proposed nonlinear currents/load torque observer and the controller.
	%
	%
	%The last validation is about the angular velocity tracking performances, and it is shown in Fig.~\ref{fig:vel_comp_Exp}.
	Figure~\ref{fig:vel_comp_Exp} shows the velocity tracking performances.
	%
	%	The red line indicates the velocity tracking performance without the load torque estimation technique, and the blue line indicates the velocity tracking performance with the load torque estimation.
	%
	We expected that the inner-loop disturbance, which has fast dynamics, will degrade the tracking performances in the transient region. And the proposed nonlinear outer/inner-loop N-PI-DOB enhance the tracking performances in the transient region.
	We observed that the tracking performance of outer/inner-loop N-PI-DOB and outer-loop N-PI-DOB are similar in steady-state region. However, they are different in a transient region such as acceleration duration, 13.5$\sim$13.6  sec. We confirmed that the proposed method enhances the motion tracking performances in transient region with the inner-loop disturbance estimation/compensation logics.
	The experimental results validates the effectiveness of the proposed nonlinear PMSM controller.
	%
	
	%%%%%%%%%%%%%%%%%%%%%%%%%%%%%%%%%%%%%%%%%%%%%%%%%%%%%%%%%%%%%%%%%%%%%%%%%
	%

	%%%%%%%%%%%%%%%%%%%%%%%%%%%% FFT
	%	\begin{figure}
	%		\centering
	%		\subfigure[FFT Result of $i_\beta$]{
	%			\includegraphics[width=0.95\linewidth]{Fig/RPM_M22_TLcomp_500to1000_ibe_FFT.eps}
	%			\label{fig:RPM_M21_TLcomp_ibe_FFT}
	%		}
	%		\subfigure[FFT Result of $i_\alpha$]{
	%			\includegraphics[width=0.95\linewidth]{Fig/RPM_M22_TLcomp_500to1000_ial_FFT.eps}
	%			\label{fig:RPM_M21_TLcomp_ial_FFT}
	%		}
	%		\subfigure[FFT Result of $e_\omega$]{z	
	%			\includegraphics[width=0.95\linewidth]{Fig/RPM_M22_TLcomp_500to1000_ew_FFT.eps}
	%			\label{fig:RPM_M21_TLcomp_ew_FFT}
	%		}
	%		\caption{FFT Analysis}
	%		\label{fig:FFT_results}
	%	\end{figure}

	%\begin{figure}
	%	\centering
	%	\subfigure[Zoom of $e_\omega$]{
	%		\includegraphics[width=0.95\linewidth]{Fig/RPM_M22_TLcomp_500to1000_ew_zoom.eps}
	%		\label{fig:RPM_M21_TLcomp_ew_zoom}
	%	}
	%	\subfigure[FFT Result of $e_\omega$]{
	%		\includegraphics[width=0.95\linewidth]{Fig/RPM_M22_TLcomp_500to1000_ew_FFT.eps}
	%		\label{fig:RPM_M21_TLcomp_ew_FFT}
	%	}
	%	\caption{Analysis of $e_\omega$}
	%	\label{fig:FFT_ew}
	%\end{figure}
	
	%\begin{figure}
	%	\centering
	%	\subfigure[Zoom of $\tau_L$]{
	%		\includegraphics[width=0.95\linewidth]{Fig/RPM_M22_TLcomp_500to1000_Tload_zoom.eps}
	%		\label{fig:RPM_M21_TLcomp_Tload_zoom}
	%	}
	%	\subfigure[FFT Result of $\tau_L$]{
	%		\includegraphics[width=0.95\linewidth]{Fig/RPM_M22_TLcomp_500to1000_Tload_FFT.eps}
	%		\label{fig:RPM_M21_TLcomp_Tload_FFT}
	%	}
	%	\caption{Analysis of $\tau_L$}
	%	\label{fig:FFT_Tload}
	%\end{figure}

	\section{Conclusion}
%	This paper presented the precision motion tracking control of SPMSM with  N-PI-DOBs. Firstly, we introduced the N-PI-DOB to estimate unknown load torque and torque modulation to stabilize the mechanical motion tracking system. Then, we showed that the motion tracking error dynamics of the SPMSM are perturbed by load torque estimation errors and currents tracking errors. Further, due to inverter nonlinearity, tracking error dynamics of currents have disturbances. We designed N-PI-DOB to perform robust motion control for estimation inner-loop disturbances and applied Lyapunov-based nonlinear currents control. Finally, we analyzed the stability of the proposed SPMSM motion controller. To validate the proposed N-PI-DOBs based motion controller, we performed the comparative studies and showed the effectiveness of the proposed control/estimation law with experiment results.

	This paper presented the precision motion tracking control with a new nonlinear proportional-integral disturbance observer (N-PI-DOB) for a surface-mounted permanent magnet synchronous motor (SPMSM). Firstly, we introduced the torque modulation technique and presented the N-PI-DOB for load torque estimation to control the desired motion of SPMSM. Then, we showed that the motion tracking error dynamics of the SPMSM can be represented in the form of a 3-cascade system, including a mechanical motion system, load torque estimation system, and currents tracking system. With this 3-cascade representation, we analyzed that the currents tracking errors disturb the convergence of the motion tracking errors. Then, we presented a new nonlinear disturbance observer of inverter dynamics and show the global exponential stability. 	 We performed the experimental comparative studies of SPMSM controller and show that velocity tracking performances in rapidly speed varying region have been improved with the proposed outer/inner-loop N-PI-DOB and motion tracking controller. 
%	\bibliographystyle{IEEEtranTIE}
%	\bibliography{BIB_PowerElectronics}


	\begin{IEEEbiography}[{\includegraphics[width=1in,height=1.25in,clip,keepaspectratio]{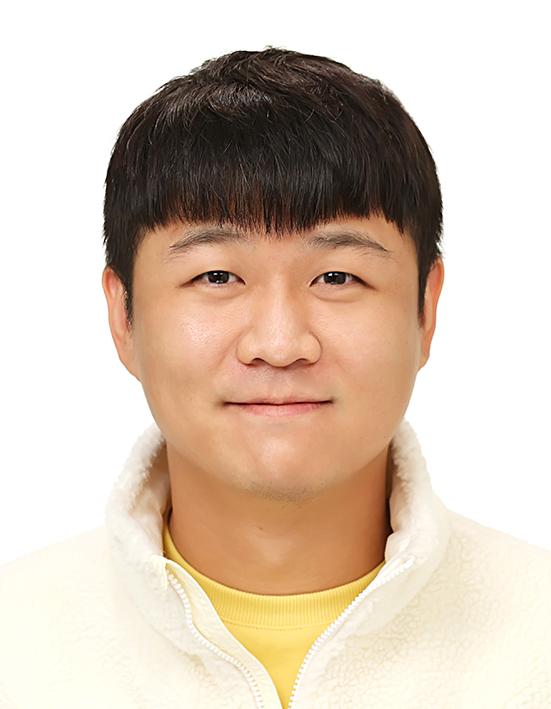}}]
		{Yong Woo Jeong} received a B.S. degree in electrical engineering from Dong-A University, Busan, South Korea, in 2016.
		He is a Ph. D. candidate with the Systems and Control Laboratory at Hanyang University, Seoul, South Korea.
		His current research interests include nonlinear control, estimators, and their applications in autonomous driving systems and power electronic systems.
		He is a member of the IEEE Control System Society, Industrial Electronics Society, Intelligent Transportation Systems Society, the Korean Society of Automotive Engineers, and the Institute of Control, Robotics, and Systems.
	\end{IEEEbiography}
	\begin{IEEEbiography}[{\includegraphics[width=1in,height=1.25in,clip,keepaspectratio]{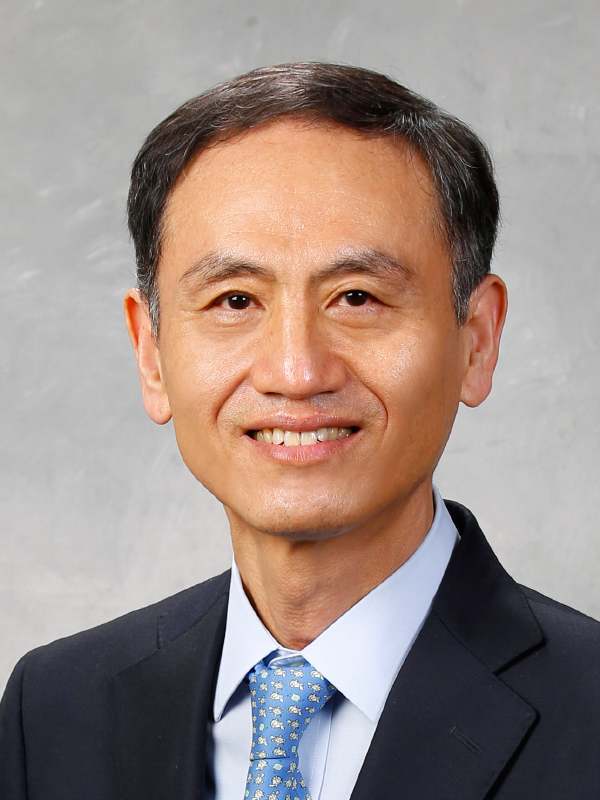}}]
		{Chung Choo Chung} (S'91-M'93) received his B.S. and M.S. degrees in electrical engineering from Seoul National University, Seoul, South Korea, and his Ph.D. degree in electrical and computer engineering from the University of Southern California, Los Angeles, CA, USA, in 1993. From 1994 to 1997, he was with the Samsung Advanced Institute of Technology, Korea. In 1997, he joined the Faculty of Hanyang University, Seoul, South Korea. Before joining Hanyang Univeristy, he was appointed as one of  21 century leaders by Samsung Group and finished Samsung Advanced Manager Program given by Wharton Business School at the University of Pennsylvania in 1996.
		%
		%Dr. Chung was an associate editor for the Asian Journal of Control from 2000 to 2002 and an founding editor for the International Journal of Control, Automation and Systems and an editor from 2003 to 2005.
		%He served as associate editor for various international conferences, such as the IEEE Conference on Decision and Control (CDC), the American Control Conferences, the IEEE Intelligent Vehicles Symposium, and the Intelligent Transportation Systems Conference. He was a guest editor for a special issue on advanced servo control for emerging data storage systems published by the IEEE Transactions on Control System Technologies (TCST), 2012 and also a guest editor for the IEEE Intelligent Transportation Systems Magazine, 2017. He was an associate editor for TCST from 2013 to 2018, the IEEE Transactions on Intelligent Transportation Systems from 2016 to 2018.
		%And he is now an associate editor for IFAC Mechatronics. He was an program co-chair of ICCAS-SICE 2009, Fukuoka, Japan, an organizing chair for the International Conference on Control, Automation and Systems (ICCAS) 2011, KINTEX, Korea and a program co-chair of the 2015 IEEE Intelligent Vehicles Symposium, COEX, Korea.
		%
		He was a general chair of ICCAS 2019, Jeju ICC, Korea and a general chair of IEEE CDC 2020, Jeju ICC, Korea.   He  was the 2019 President of the Institute of Control, Robotics and Systems (ICROS), Korea. He is a member of the National Academy
		of Engineering of Korea (NAEK). 		
	\end{IEEEbiography}
\end{document}